%% file: main8.tex
\PassOptionsToPackage{unicode}{hyperref}
\PassOptionsToPackage{hyphens}{url}
\PassOptionsToPackage{dvipsnames,svgnames,x11names}{xcolor}
\documentclass[
  12pt]{article}
\usepackage{hyphenat}

\usepackage{amsmath,amssymb}
\usepackage{microtype}
\DeclareMathOperator{\Var}{Var}
\usepackage{mathtools}
\usepackage[shortlabels]{enumitem}
\newcommand{\E}{\mathbb{E}}
\newcommand{\CV}{\mathrm{CV}}

\newcommand{\Hone}{\mathcal{H}_1}
\newcommand{\Htwo}{\mathcal{H}_2}

\usepackage{iftex}
\ifPDFTeX
  \usepackage[T1]{fontenc}
  \usepackage[utf8]{inputenc}
  \usepackage{textcomp} 
\else 
  \usepackage{unicode-math}
  \defaultfontfeatures{Scale=MatchLowercase}
  \defaultfontfeatures[\rmfamily]{Ligatures=TeX,Scale=1}
\fi
\usepackage{lmodern}
\usepackage{amsthm}

\ifPDFTeX\else  
\fi
\IfFileExists{upquote.sty}{\usepackage{upquote}}{}
\IfFileExists{microtype.sty}{
  \usepackage[]{microtype}
  \UseMicrotypeSet[protrusion]{basicmath} 
}{}
\makeatletter
\@ifundefined{KOMAClassName}{
  \IfFileExists{parskip.sty}{%
    \usepackage{parskip}
  }{
    \setlength{\parindent}{0pt}
    \setlength{\parskip}{6pt plus 2pt minus 1pt}}
}{
  \KOMAoptions{parskip=half}}
\makeatother
\usepackage{xcolor}
\makeatletter
\ifx\paragraph\undefined\else
  \let\oldparagraph\paragraph
  \renewcommand{\paragraph}{
    \@ifstar
      \xxxParagraphStar
      \xxxParagraphNoStar
  }
  \newcommand{\xxxParagraphStar}[1]{\oldparagraph*{#1}\mbox{}}
  \newcommand{\xxxParagraphNoStar}[1]{\oldparagraph{#1}\mbox{}}
\fi
\ifx\subparagraph\undefined\else
  \let\oldsubparagraph\subparagraph
  \renewcommand{\subparagraph}{
    \@ifstar
      \xxxSubParagraphStar
      \xxxSubParagraphNoStar
  }
  \newcommand{\xxxSubParagraphStar}[1]{\oldsubparagraph*{#1}\mbox{}}
  \newcommand{\xxxSubParagraphNoStar}[1]{\oldsubparagraph{#1}\mbox{}}
\fi
\makeatother

\newtheoremstyle{break}
  {\dimexpr\baselineskip+\parskip\relax}
  {\dimexpr\baselineskip+\parskip\relax}
  {\itshape}
  {}
  {\bfseries}
  {.}
  {\newline}
  {}
\theoremstyle{break}
\newtheorem{theorem}{Theorem}[section]
\newtheorem{lemma}[theorem]{Lemma}
\newtheorem{corollary}[theorem]{Corollary}
\newtheorem{proposition}[theorem]{Proposition}
\newtheorem{remark}[theorem]{Remark}

\newcounter{algobox}
\renewcommand{\thealgobox}{\arabic{algobox}}

\usepackage{longtable,booktabs,array}
\usepackage{calc} 
\usepackage{etoolbox}
\makeatletter
\patchcmd\longtable{\par}{\if@noskipsec\mbox{}\fi\par}{}{}
\makeatother
\IfFileExists{footnotehyper.sty}{\usepackage{footnotehyper}}{\usepackage{footnote}}
\makesavenoteenv{longtable}
\usepackage{graphicx}
\makeatletter
\def\maxwidth{\ifdim\Gin@nat@width>\linewidth\linewidth\else\Gin@nat@width\fi}
\def\maxheight{\ifdim\Gin@nat@height>\textheight\textheight\else\Gin@nat@height\fi}
\makeatother
\setkeys{Gin}{width=\maxwidth,height=\maxheight,keepaspectratio}
\makeatletter
\def\fps@figure{htbp}
\makeatother

\makeatletter
\def\fps@figure{htbp}
\makeatother

\usepackage{placeins} 

\makeatletter
\@ifpackageloaded{caption}{}{\usepackage{caption}}
\AtBeginDocument{%
\ifdefined\contentsname
  \renewcommand*\contentsname{Table of contents}
\else
  \newcommand\contentsname{Table of contents}
\fi
\ifdefined\listfigurename
  \renewcommand*\listfigurename{List of Figures}
\else
  \newcommand\listfigurename{List of Figures}
\fi
\ifdefined\listtablename
  \renewcommand*\listtablename{List of Tables}
\else
  \newcommand\listtablename{List of Tables}
\fi
\ifdefined\figurename
  \renewcommand*\figurename{Figure}
\else
  \newcommand\figurename{Figure}
\fi
\ifdefined\tablename
  \renewcommand*\tablename{Table}
\else
  \newcommand\tablename{Table}
\fi
}
\@ifpackageloaded{float}{}{\usepackage{float}}
\floatstyle{ruled}
\@ifundefined{c@chapter}{\newfloat{codelisting}{h}{lop}}{\newfloat{codelisting}{h}{lop}[chapter]}
\floatname{codelisting}{Listing}

\makeatother
\makeatletter
\@ifpackageloaded{caption}{}{\usepackage{caption}}
\@ifpackageloaded{subcaption}{}{\usepackage{subcaption}}
\makeatother

\ifLuaTeX
  \usepackage{selnolig}  
\fi
\usepackage[]{natbib}
\usepackage{bookmark}
\usepackage{comment}

\IfFileExists{xurl.sty}{\usepackage{xurl}}{} 
\hypersetup{
  pdftitle={Optimality of the Half‑Order Exponent in the Turing–Good Identities for Bayes Factors},
  pdfauthor={Kensuke Okada},
  pdfkeywords={3 to 6 keywords, that do not appear in the title},
  colorlinks=true,
  linkcolor={blue},
  filecolor={Maroon},
  citecolor={Blue},
  urlcolor={Blue},
  pdfcreator={LaTeX via pandoc}}

\newcommand{\anon}{1}
\providecommand{\half}{\frac{1}{2}}

\begin{document}

\def\spacingset#1{\renewcommand{\baselinestretch}%
{#1}\small\normalsize} \spacingset{1}


\if1\anon
{
  \title{\bf Optimality of the Half‑Order Exponent in the Turing–Good Identities for Bayes Factors}
  \author{Kensuke Okada\hspace{.2cm}\\
    Graduate School of Education, The University of Tokyo}
  \maketitle
} \fi

\if0\anon
{
  \bigskip
  \bigskip
  \bigskip
  \begin{center}
    {\LARGE\bf Title}
\end{center}
  \medskip
} \fi

\bigskip
\begin{abstract}
Bayes factors are widely computed by Monte Carlo, yet heavy-tailed sampling distributions can make numerical validation unreliable. The Turing--Good identities provide exact moment equalities for powers of a Bayes factor (a density ratio). When these identities are used as Good-check diagnostics, the power choice becomes a statistical design parameter. We develop a nonasymptotic variance theory for Monte Carlo evaluation of the identities and show that the half-order (square-root) power is uniquely minimax-stable: it equalizes variability across the two model orientations and is the only choice that guarantees finite second moments in a distribution-free worst-case sense over all mutually absolutely continuous model pairs. This yields a balanced two-sample half-order diagnostic that is symmetric in model labeling and has a uniform variance bound at fixed computational budget; in small-overlap regimes it is guaranteed to be no less efficient than the standard one-sided Turing check. Simulations for binomial Bayes factor workflows illustrate stable finite-sample behavior and sensitivity to simulator--evaluator mismatches. We further connect the half-order overlap viewpoint to stable primitives for normalizing-constant ratios and importance-sampling degeneracy summaries.
\end{abstract}


\newpage
\spacingset{1.8} 

\section{Introduction}
\label{sec:intro}

Bayes factors provide a canonical Bayesian measure of relative evidence between two statistical models.
Given rival hypotheses $\mathcal{H}_1$ and $\mathcal{H}_2$ with marginal likelihoods (prior predictives)
$p_1(x)$ and $p_2(x)$, the Bayes factor
\[
B(x)\;:=\;\frac{p_1(x)}{p_2(x)}
\]
updates prior odds into posterior odds via multiplication \citep{KassRaftery1995}.
As \citet[][p. 558]{HeckEtAl2023} note, ``The last 25 years have shown a steady increase in attention for the Bayes factor as a tool for hypothesis evaluation and model selection.'' 
 Bayes factors are routinely used as decision statistics
(e.g., threshold rules that declare $\mathcal{H}_1$ when $B(x)$ is large), which raises both computational and
frequentist-calibration questions \citep{Dickey1971,KassRaftery1995}. Bayes factors are used not only for hypothesis tests but also to test alternative substantive theories against each other \citep{LeeWagenmakers2013}.
While asymptotic approximations such as BIC connect Bayes factors to log-likelihood differences in regular models
\citep{schwarz1978}, modern applications often hinge on reliable \emph{finite-sample} computations of marginal likelihoods. 

Computing marginal likelihoods---and more generally, estimating ratios of normalizing constants---is notoriously delicate.
Many practical algorithms reduce these problems to Monte Carlo evaluation of expectations of a density ratio
(the Radon--Nikodym derivative), using importance sampling, bridge/path sampling, and related free-energy estimators
\citep{Bennett1976,MengWong1996,ShirtsChodera2008,Gronau2017}.
Such methods can be highly effective when the relevant distributions overlap substantially.
However, under tail mismatch the ratio can develop extreme weights that dominate Monte Carlo averages,
leading to severe variance inflation, highly skewed sampling distributions, and unstable standard-error assessments.

Accordingly, there is a growing emphasis on \emph{workflow-level} diagnostics and rigorous validation of Bayes factor
implementations \citep{SchadEtAl2023}.
A remarkably general lens on the sampling behavior of density ratios is provided by the classic Turing--Good moment identity \citep{Good1985,jacod2013limit}:
\[
\E_{\mathcal{H}_2}\!\left[B(X)^t\right]
\;=\;
\E_{\mathcal{H}_1}\!\left[B(X)^{t-1}\right].
\]
Recently, \citet{Sekulovski2024} advocated leveraging selected instances of this identity, specifically
$t=1$ (the ``Turing'' identity $\E_{\mathcal{H}_2}[B]=1$) and $t=2$
---as practical \emph{Good checks}
for validating numerical Bayes factor computations.

When these exact population identities are repurposed as Monte Carlo diagnostics, the exponent $t$ becomes a critical
statistical \emph{design parameter}.
The identities fix means, but numerical reliability is governed by second moments.
Because $B(X)$ is often highly skewed and heavy-tailed, natural integer-order choices like the $t=1$ or $t=2$
checks can exhibit enormous or even infinite variance for perfectly valid, mutually absolutely continuous model pairs.
This exposes a methodological paradox: the very regimes where validation is most needed (weak overlap and heavy tails)
are exactly those where the diagnostic statistics themselves can become ill-posed in finite Monte Carlo runs.

This paper resolves the design paradox by studying exponent selection from an explicitly computation-aware perspective.
We ask:
\emph{which exponent yields a moment identity that can be verified stably, in a distribution-free sense?}

We show that the half-order exponent $t=\tfrac12$ 
is mathematically distinguished.
Among the continuum of valid identities, the half-order is the \emph{unique} exponent that equalizes the natural variances
under $\mathcal{H}_1$ and $\mathcal{H}_2$, and it uniquely yields a distribution-free, finite second-moment guarantee for
the per-draw diagnostic building blocks across all mutually absolutely continuous pairs.

Our main contributions are threefold:

\begin{enumerate}
\item \textbf{Exact variance theory and unique minimax stability (Section~2).}
We derive exact nonasymptotic variance expressions for Monte Carlo evaluation of the expectationfrom either side of the
Turing--Good identity in terms of shifted overlaps.
At half order, $B^{1/2}$ under $\mathcal{H}_2$ and $B^{-1/2}$ under $\mathcal{H}_1$ have equal finite variance.
Conversely, for every $t\neq\tfrac12$ there exist mutually absolutely continuous pairs for which the worst-side variance
diverges.

\item \textbf{Robust Good checks at matched cost (Section~3).}
We propose a symmetric, two-sided half-order diagnostic that resolves the main bottlenecks of integer-moment Good checks.
Under a fixed budget of Bayes factor evaluations, the balanced two-sample half-order discrepancy has a bounded variance, uniformly over mutually absolutely continuous pairs.
We also provide an analytic matched-budget comparison showing that, in the small-overlap regime,
the half-order two-sided check is guaranteed to be no less efficient (in variance) than the one-sided $t=1$ Turing check.
Simulation studies illustrate stable finite-sample behavior and sensitivity to simulator--evaluator mismatches.

\item \textbf{Broader implications for ratio estimation and weight degeneracy
(Section~4).}
We show that the same half-order viewpoint extends beyond Good checks to core computational tasks:
a canonical geometric bridge identity for normalizing-constant ratios whose building blocks have distribution-free finite
second moments and stable overlap-guided summaries of importance-weight concentration that do not rely on fragile second moments
\end{enumerate}


Section~2 establishes the exact variance structure of the Turing--Good identities and proves the unique minimax stability
of the half-order exponent.
Section~3 develops the balanced two-sided half-order diagnostic, provides matched-budget variance comparisons, and presents
simulation studies.
Section~4 explores broader implications for normalizing-constant ratios and importance sampling diagnostics
.
Section~5 concludes with a discussion. 

\section{Turing--Good Identities, Variance Structure, and Half-Order Optimality}
\label{sec:TG-variance-halforder}

\subsection{Setup and basic properties}


Let $(\mathcal X,\mathcal A,\mu)$ be a measurable space with a $\sigma$-finite dominating measure $\mu$.
We consider two competing models (or hypotheses) $\mathcal{H}_1$ and $\mathcal{H}_2$ with \emph{marginal likelihood}
(prior predictive) densities $p_1$ and $p_2$ with respect to~$\mu$.
Let $P_j$ be the induced probability measures, $P_j(A):=\int_A p_j\,d\mu$.
Throughout we assume \emph{mutual absolute continuity} $P_1\ll P_2$ and $P_2\ll P_1$; equivalently,
\begin{equation}
p_1(x)=0 \iff p_2(x)=0 \qquad \text{for $\mu$-a.e.\ $x\in\mathcal X$}.
\label{eq:mutual-ac}
\end{equation}
so that the Bayes factor is well-defined $P_2$-a.s.\ (and $P_1$-a.s.):
\[
B(x):=\frac{p_1(x)}{p_2(x)}\quad\text{for }p_2(x)>0,
\]
with an arbitrary (e.g., $1$) value on $\{p_2(x)=0\}$, which is $P_1$- and $P_2$-null.
For readability, we write $\int f(x)\,dx$ as shorthand for $\int f\,d\mu$.

For $t\in\mathbb R$, define
\begin{equation}
I(t)\;:=\;\int p_1(x)^{\,t}\,p_2(x)^{\,1-t}\,dx,
\label{eq:def-It}
\end{equation}
the Hellinger integral of order~$t$ \citep[p.~228]{jacod2013limit}.

We denote the effective domain (which depends on the pair $(p_1,p_2)$) by
\[
D = D(p_1,p_2)\;:=\;\{t\in\mathbb R:\ I(t)<\infty\}.
\]
Unless stated otherwise, whenever we write $\E_{\mathcal{H}_2}[B^t]$ or $\E_{\mathcal{H}_1}[B^{t-1}]$ we implicitly assume $t\in D$.
Note that $t\in D$ does not in general imply $2t\in D$ or $2t-1\in D$; these additional conditions 
govern the finiteness of the second moments (and hence whether the variances in Lemma \ref{lem:var-identities} are finite).
For $t\in (0,\infty)\cap D\setminus\{1\}$, the quantity
\[
  D_t(p_1\|p_2) := \frac{1}{t-1}\log I(t)
\]
is the order-$t$ Rényi divergence from $p_1$ to $p_2$.
At $t=1$, interpret by continuity:
\[
  \lim_{t\to 1} \frac{1}{t-1}\log I(t) = D_{\mathrm{KL}}(p_1\|p_2).
\]

We write $\mathbb{P}_{\mathcal H_j}$ and $\E_{\mathcal H_j}$ for probability/expectation under the law with density $p_j$
(equivalently, $\mathbb{P}_{p_j}$ and $\E_{p_j}$).

\begin{theorem}[Turing--Good moment identity; \citealt{Good1985}]
\label{thm:TG}
For any $t\in D$,
\[
  \E_{\mathcal{H}_1}\!\left[B^{\,t-1}\right]
  \ =\
  \E_{\mathcal{H}_2}\!\left[B^{\,t}\right]
  \ =\ I(t).
\]
\end{theorem}

\begin{proof}
By definition of $B$,
\[
\E_{\mathcal{H}_1}\!\left[B^{\,t-1}\right]
=\int \Bigl(\frac{p_1}{p_2}\Bigr)^{t-1}p_1\,dx
=\int p_1^{\,t}p_2^{\,1-t}\,dx
=I(t),
\]
and similarly
\[
\E_{\mathcal{H}_2}\!\left[B^{\,t}\right]
=\int \Bigl(\frac{p_1}{p_2}\Bigr)^{t}p_2\,dx
=\int p_1^{\,t}p_2^{\,1-t}\,dx
=I(t).
\]
\end{proof}

\begin{corollary}[Turing's theorem]
\label{cor:Turing}
At $t=1$ in Theorem~\ref{thm:TG}, $\E_{\mathcal{H}_2}[B]=1$.
\end{corollary}

\begin{lemma}[Basic bounds and equality characterization]
\label{lem:basic-bounds}
\begin{enumerate}
(i) $I(0)=I(1)=1$.
(ii) For every $t\in[0,1]$, $I(t)\le 1$ (in particular, $[0,1]\subset D$).
(iii) For every $t\notin[0,1]$ with $t\in D$, $I(t)\ge 1$.
\end{enumerate}
If $t\notin\{0,1\}$, then equality in (ii) or (iii) holds if and only if $p_1(x)=p_2(x)$ $\mu$-a.e.
\end{lemma}

\begin{proof}
(i) follows from $I(0)=\int p_2\,dx=1$ and $I(1)=\int p_1\,dx=1$.
For (ii), if $0<t<1$ then the weighted AM--GM inequality yields
$p_1(x)^{t}p_2(x)^{1-t}\le t\,p_1(x)+(1-t)\,p_2(x)$ pointwise; integrating both sides gives $I(t)\le 1$.
For (iii), let $f_t(u)=u^t$ on $(0,\infty)$.
For $t<0$ or $t>1$, $f_t$ is strictly convex, so Jensen's inequality gives
\[
I(t)=\E_{\mathcal{H}_2}[B^{t}] \ \ge\ \{\E_{\mathcal{H}_2}[B]\}^{t} \ =\ 1,
\]
where we used Corollary~\ref{cor:Turing} in the last equality.
If $t\notin\{0,1\}$, equality in Jensen holds if and only if $B$ is $\mathcal{H}_2$-a.s.\ constant.
Since $\E_{\mathcal{H}_2}[B]=1$, this forces $B\equiv 1$ $\mathcal{H}_2$-a.s., and hence $p_1=p_2$ $\mu$-a.e.
\end{proof}

\paragraph*{Remark (Domain and finiteness on $[0,1]$).}
Lemma~\ref{lem:basic-bounds}(ii) implies that $[0,1]\subset D$ and that $I(t)\in[0,1]$ for all $t\in[0,1]$.
In particular, the half-order overlap $\rho=I(\half)$ is always well-defined and finite.
Under mutual absolute continuity~\eqref{eq:mutual-ac}, we also have $\rho>0$
(and $\rho=1$ if and only if $p_1=p_2$ $\mu$-a.e.).

A quantity that will play a central role is the \emph{half-order} overlap
\begin{equation}
\rho \;:=\; I(\half) \;=\;\int \sqrt{p_1(x)p_2(x)}\,dx \ \in\ (0,1],
\label{eq:rho-def}
\end{equation}
the Hellinger affinity (Bhattacharyya coefficient) between the two marginal likelihoods.

\subsection{Variance identities and a CGF viewpoint}

The Turing--Good moment identity fixes the means of the transformed Bayes factors
$B^{t-1}$ under $\mathcal{H}_1$ and $B^t$ under $\mathcal{H}_2$.
For exponent selection, the next natural object is their dispersion.

\paragraph*{Convention (Second moments and extended variances).}
We allow variances to take values in $[0,\infty]$, interpreting $\Var(\cdot)=+\infty$ when the corresponding second moment is infinite; equivalently,
 $\Var_{\mathcal{H}_2}(B^{t})<\infty$ exactly when $2t\in D$ and $\Var_{\mathcal{H}_1}(B^{t-1})<\infty$ exactly when $2t-1\in D$.

\begin{lemma}[Variance identities]
\label{lem:var-identities}
For any $t\in D$,
\begin{equation}
\Var_{\mathcal{H}_1}\!\bigl(B^{t-1}\bigr)=I(2t-1)-I(t)^2,
\qquad
\Var_{\mathcal{H}_2}\!\bigl(B^{t}\bigr)=I(2t)-I(t)^2,
\end{equation}
where the right-hand sides are interpreted in $[0,\infty]$ according to the preceding convention.
\label{eq:var-id}
\end{lemma}

\begin{proof}
Fix $t\in D$, so $I(t)<\infty$ and hence $\E_{\Htwo}[B^t]=\E_{\Hone}[B^{t-1}]=I(t)$.
For the $\Htwo$ side, using $\Var(X)=\E[X^2]-\{\E[X]\}^2$ with the convention that $\Var(X)=+\infty$
when $\E[X^2]=+\infty$,
\[
\Var_{\Htwo}(B^t)=\E_{\Htwo}[B^{2t}] - \{\E_{\Htwo}[B^t]\}^2.
\]
Moreover,
\[
\E_{\Htwo}[B^{2t}]
= \int \left(\frac{p_1}{p_2}\right)^{2t} p_2\,d\mu
= \int p_1^{2t} p_2^{1-2t}\,d\mu
= I(2t)\in[0,\infty].
\]
Therefore $\Var_{\Htwo}(B^t)=I(2t)-I(t)^2$ in $[0,\infty]$.
The $\Hone$ side is analogous:
\[
\E_{\Hone}[B^{2(t-1)}]
= \int \left(\frac{p_1}{p_2}\right)^{2(t-1)} p_1\,d\mu
= \int p_1^{2t-1} p_2^{2-2t}\,d\mu
= I(2t-1),
\]
hence $\Var_{\Hone}(B^{t-1})=I(2t-1)-I(t)^2$.
\end{proof}

\begin{corollary}[Half-order equalization]
\label{cor:half-order-eq}
At $t=\half$, the two variances coincide and are always finite:
\[
\Var_{\mathcal{H}_1}\!\bigl(B^{-1/2}\bigr)
=\Var_{\mathcal{H}_2}\!\bigl(B^{1/2}\bigr)
=1-\rho^2.
\]
\end{corollary}

\begin{proof}
Apply Lemma~\ref{lem:var-identities} at $t=\half$ and use Lemma~\ref{lem:basic-bounds}(i), i.e., $I(0)=I(1)=1$, together with~\eqref{eq:rho-def}.
\end{proof}

Corollary~\ref{cor:half-order-eq} already hints at the special role of $t=\half$:
the relevant second moments are anchored at $I(0)$ and $I(1)$, which are identically~$1$
for \emph{all} pairs $(p_1,p_2)$, whereas for generic $t$ one must control $I(2t)$ and $I(2t-1)$,
which may fail to be finite.

To formalize a minimax statement, we use a cumulant generating function viewpoint.

\begin{lemma}[CGF of $\log B$ and log-convexity of $I$]
\label{lem:cgf}
Let $Y:=\log B$ and $\phi(t):=\log I(t)$.
For $t\in D$, $I(t)=\E_{\mathcal{H}_2}[e^{tY}]$, hence $\phi$ is the cumulant generating function of $Y$
under $\mathcal{H}_2$.
In particular, $I$ is log-convex on $D$ (equivalently, $\phi$ is convex on $D$), and $D$ is a convex set.
Moreover, if $p_1\neq p_2$ $\mu$-a.e., then $\phi$ is strictly convex on $\mathrm{int}(D)$.
\end{lemma}

\begin{proof}
The identity $I(t)=\E_{\mathcal{H}_2}[e^{tY}]$ is immediate from Theorem~\ref{thm:TG} and the definition of $Y$.
For log-convexity, let $t_1,t_2\in D$ and $\lambda\in(0,1)$.
By H\"older's inequality,
\[
I\bigl(\lambda t_1+(1-\lambda)t_2\bigr)
=\E_{\mathcal{H}_2}\!\left[e^{(\lambda t_1+(1-\lambda)t_2)Y}\right]
=\E_{\mathcal{H}_2}\!\left[e^{\lambda t_1 Y}\,e^{(1-\lambda)t_2 Y}\right]
\le I(t_1)^{\lambda}\,I(t_2)^{1-\lambda}.
\]
Taking logarithms yields convexity of $\phi$, and in particular the left-hand side is finite, so $D$ is convex.
If $p_1\neq p_2$ $\mu$-a.e., then $Y$ is nondegenerate under $\mathcal{H}_2$, and strictness of H\"older
(for $t_1\neq t_2$ in $\mathrm{int}(D)$) implies strict convexity of $\phi$ on $\mathrm{int}(D)$.
\end{proof}

\subsection{A minimax principle for choosing the exponent}

For $t\in D$, define
\[
  R_1(t):=\Var_{\mathcal{H}_1}(B^{t-1})\in[0,\infty],\qquad
  R_2(t):=\Var_{\mathcal{H}_2}(B^{t})\in[0,\infty].
\]
Since $I(t)<\infty$ for $t\in D$, Lemma~\ref{lem:var-identities} together with the convention on extended variances implies that $R_2(t)<\infty$ if and only if $2t\in D$, and $R_1(t)<\infty$ if and only if $2t-1\in D$.
We define the worst-side risk
\[
  R(t):=\max\{R_1(t),R_2(t)\}\in[0,\infty].
\]

\paragraph*{Notation (Worst-case comparisons).}
Define the class of mutually absolutely continuous density pairs
\[
\mathcal P_{\mathrm{ac}}
:=\Bigl\{(p_1,p_2):\ p_1,p_2 \text{ are }\mu\text{-densities on }(\mathcal X,\mathcal A)\text{ satisfying \eqref{eq:mutual-ac}}\Bigr\}.
\]
All ``worst-case'' suprema below (e.g., $\sup_{(p_1,p_2)\in\mathcal P_{\mathrm{ac}}}$) are taken over this class.

We now show that $t=\half$ is the unique minimax choice for $R(t)$ and, moreover, the only exponent
whose worst-case risk is uniformly finite over mutually absolutely continuous pairs $(p_1,p_2)$.

\begin{theorem}[Half-order is minimax and uniquely worst-case stable]
\label{thm:minimax}
Assume $p_1\neq p_2$ $\mu$-a.e.
\begin{enumerate}
\item[(i)] \textnormal{(Pairwise minimaxity and uniqueness.)}
For every $t\in D$,
\[
R(t)\ \ge\ 1-\rho^2,
\]
with equality if and only if $t=\half$.
In particular, $t=\half$ uniquely minimizes $R(t)$ over $t\in D$, and
$R(\half)=1-\rho^2$.

\item[(ii)] \textnormal{(Worst-case divergence away from half-order.)}
For every $t\neq \half$, there exists a mutually absolutely continuous pair $(p_1,p_2)$ for which
$I(t)<\infty$ but $R(t)=+\infty$.
In particular, for every fixed $t\neq\half$,
\[
\sup\Big\{R_{p_1,p_2}(t):\ (p_1,p_2)\in\mathcal P_{\mathrm{ac}},\ I(t)<\infty\Big\}=+\infty,
\]
whereas for every $(p_1,p_2)\in\mathcal P_{\mathrm{ac}}$,
\[
R_{p_1,p_2}(\half)=1-\rho(p_1,p_2)^2<\infty.
\]

\end{enumerate}
\end{theorem}

\begin{remark}[Degenerate equal-model case]
If $p_1=p_2$ $\mu$-a.e., then $B\equiv 1$ and $I(t)=1$ for all $t$, so $R_1(t)=R_2(t)=0$ and exponent selection is immaterial.
The assumption $p_1\neq p_2$ in Theorem~\ref{thm:minimax} is used only to obtain strict inequalities and uniqueness.
\end{remark}

\begin{proof}
Write $\phi(t)=\log I(t)$ as in Lemma~\ref{lem:cgf}.
For $t\ge \half$ (with $2t\in D$) define $g_+(t):=\phi(2t)-2\phi(t)$ and for $t\le \half$ (with $2t-1\in D$) define $g_-(t):=\phi(2t-1)-2\phi(t)$.
Since $\phi$ is convex, its right-derivative $\phi'_+(t)$ exists on $\mathrm{int}(D)$ and is nondecreasing.
Thus for $t\ge \half$ (whenever $t,2t\in \mathrm{int}(D)$),
\[
g'_+(t)=2\phi'_+(2t)-2\phi'_+(t)\ge 0.
\]
Similarly, for $t\le \half$ (whenever $t,2t-1\in \mathrm{int}(D)$),
\[
g'_-(t)=2\phi'_+(2t-1)-2\phi'_+(t)\le 0.
\]
Therefore $g_+$ is nondecreasing on $\{t\in D:\ t\ge \half,\ 2t\in D\}$ and $g_-$ is nonincreasing on $\{t\in D:\ t\le \half,\ 2t-1\in D\}$.
Exponentiating, for all $t$ such that the displayed quantities are finite,
\begin{align}
\frac{I(2t)}{I(t)^2}
= e^{g_+(t)}
\ \ge\
e^{g_+(\half)}
=\frac{I(1)}{I(\half)^2}
=\rho^{-2}
\qquad & (t\ge \half),
\label{eq:ratio-right}\\[2pt]
\frac{I(2t-1)}{I(t)^2}
= e^{g_-(t)}
\ \ge\
e^{g_-(\half)}
=\frac{I(0)}{I(\half)^2}
=\rho^{-2}
\qquad & (t\le \half).
\label{eq:ratio-left}
\end{align}

\noindent\emph{Proof of (i).}
We split by the location of $t$ relative to~$\half$.

\emph{Case $t>\half$.}
If $I(2t)=\infty$, then $\E_{\mathcal{H}_2}[B^{2t}]=I(2t)=\infty$, hence $\Var_{\mathcal{H}_2}(B^t)=\infty$ and therefore $R(t)=\infty\ge 1-\rho^2$.
Otherwise $I(2t)<\infty$, and~\eqref{eq:ratio-right} yields $I(t)^2\le \rho^2 I(2t)$.
Thus,
\[
R_2(t)=I(2t)-I(t)^2 \ \ge\ I(2t)\,(1-\rho^2).
\]
Moreover, since $2t>1$, Lemma~\ref{lem:basic-bounds}(iii) gives $I(2t)>1$ for $p_1\neq p_2$, so $R_2(t)>1-\rho^2$.

\emph{Case $t<\half$.}
If $I(2t-1)=\infty$, then $\E_{\mathcal{H}_1}[B^{2(t-1)}]=I(2t-1)=\infty$, hence $\Var_{\mathcal{H}_1}(B^{t-1})=\infty$ and therefore $R(t)=\infty\ge 1-\rho^2$.
Otherwise $I(2t-1)<\infty$, and~\eqref{eq:ratio-left} yields $I(t)^2\le \rho^2 I(2t-1)$.
Therefore,
\[
R_1(t)=I(2t-1)-I(t)^2 \ \ge\ I(2t-1)\,(1-\rho^2).
\]
Since $2t-1<0$, Lemma~\ref{lem:basic-bounds}(iii) gives $I(2t-1)>1$ for $p_1\neq p_2$, so $R_1(t)>1-\rho^2$.

\emph{Case $t=\half$.}
By Corollary~\ref{cor:half-order-eq}, $R(\half)=1-\rho^2$ and both side-specific variances coincide.

Combining the three cases proves $R(t)\ge 1-\rho^2$ for all $t\in D$, with equality only at $t=\half$.

\medskip\noindent\emph{Proof of (ii).}
Fix $t\neq \half$.
It suffices to exhibit (for each such $t$) a single mutually absolutely continuous pair $(p_1,p_2)$
for which $I(t)<\infty$ but either $I(2t)=\infty$ (when $t>\half$) or $I(2t-1)=\infty$ (when $t<\half$),
because then~\eqref{eq:var-id} gives $R(t)=\infty$.

Work on $((0,1),\mathcal B,\mu)$ with $\mu$ the Lebesgue measure and take $p_2(x)\equiv 1$.
\begin{itemize}
\item If $t>\half$, let $\gamma:=1/(2t)\in(0,1)$ and set $p_1(x)=(1-\gamma)x^{-\gamma}\mathbf 1_{(0,1)}(x)$.
Then
\[
I(s)=\int_0^1 p_1(x)^s\,dx
=(1-\gamma)^s\int_0^1 x^{-\gamma s}\,dx
=\begin{cases}
\dfrac{(1-\gamma)^s}{1-\gamma s}, & \gamma s<1,\\[6pt]
+\infty, & \gamma s\ge 1.
\end{cases}
\]
Since $\gamma t=\half<1$ while $\gamma(2t)=1$, we have $I(t)<\infty$ but $I(2t)=\infty$, so $R_2(t)=\infty$.

\item If $t<\half$, let $\gamma:=1/(1-2t)>0$ and set $p_1(x)=(\gamma+1)x^{\gamma}\mathbf 1_{(0,1)}(x)$.
Then
\[
I(s)=\int_0^1 p_1(x)^s\,dx
=(\gamma+1)^s\int_0^1 x^{\gamma s}\,dx
=\begin{cases}
\dfrac{(\gamma+1)^s}{\gamma s+1}, & \gamma s>-1,\\[6pt]
+\infty, & \gamma s\le -1.
\end{cases}
\]
Here $\gamma t=t/(1-2t)>-1$, so $I(t)<\infty$, but $\gamma(2t-1)=-1$, so $I(2t-1)=\infty$ and hence $R_1(t)=\infty$.
\end{itemize}
In both cases, the constructed $(p_1,p_2)$ are strictly positive on $(0,1)$ and therefore mutually absolutely continuous, completing the proof.
\end{proof}

\section{Good checks: two-sided half-order diagnostics and simulation studies}\label{sec:goodchecks}

This section develops the implications of Section~\ref{sec:TG-variance-halforder} for \emph{Good checks} of Bayes factor
computations \citep{Sekulovski2024}.
We first recall why moment-identity checks can be fragile in practice, due to the skewness and heavy tails of the Bayes
factor under either hypothesis.
We then formalize a two-sided diagnostic based on independently estimating both sides of the Turing--Good identity,
establishing that the half-order choice is uniquely minimax-stable in a two-sample sense and yields a favorable
matched-budget variance comparison.
Finally, we provide numerical illustrations that highlight the finite-sample stability of the half-order check and its
ability to detect implementation and prior mismatches.

\subsection{Good checks for Bayes factor computations: motivation and bottlenecks}\label{subsec:goodcheck}

The immediate motivation for this work is the \emph{Good check} proposed by
\citet{Sekulovski2024}, a practical diagnostic that leverages the
Turing--Good identities to validate numerical Bayes factor computations.
Good checks are attractive because they are broadly \emph{model-agnostic}: they apply to any pair of
marginal likelihoods $(p_1,p_2)$ defined on a common measure space, provided one can
(i) simulate from the corresponding prior predictives and (ii) evaluate the Bayes factor $B=p_1/p_2$
for each simulated dataset.
At the same time, the diagnostic inherits a fundamental difficulty from Bayes factors
themselves---namely, the skewness and heavy-tail behavior of $B=p_1/p_2$ under either model.
We summarize key bottlenecks of the original Good-check design and use them to motivate the
half-order viewpoint developed in Section~\ref{sec:TG-variance-halforder}.

\paragraph*{Existing bottleneck: Good checks are vulnerable to the tails of the Bayes factor distribution.}
The Turing--Good identities are exact population equalities
(e.g., Corollary~\ref{cor:Turing} and Theorem~\ref{thm:TG}), but Good checks must evaluate them
numerically, typically via Monte Carlo:
one simulates synthetic datasets under a designated generating model, computes Bayes factors using the
implementation under scrutiny, and then checks whether the relevant empirical averages
match the theoretical values \citep{Sekulovski2024}.
The practical difficulty is that the Bayes factor distribution under the generating model is often
highly asymmetric with heavy tails.
As a result, Monte Carlo averages can be dominated by \emph{rare extreme events}:
even though such events have tiny probability, they can induce astronomically large Bayes factor values,
and hence contribute non-negligibly (or even decisively) to the mean.

A simple illustration already appears in the canonical binomial point-null example.
When $\Htwo:\theta=1/2$ is true and one evaluates a Bayes factor favoring $\Hone$,
events such as $y=0$ or $y=1$ have probabilities $2^{-n}$ and $n2^{-n}$, respectively;
for $n=50$ these are about $8.9\times 10^{-16}$ and $4.4\times 10^{-14}$.
Such events are practically never observed in finite Monte Carlo runs, yet the Bayes factor values they
produce can be extremely large.
Consequently, the empirical mean may fail to approximate its theoretical value even when the
implementation is correct, leading to ``false alarms'' or inconclusive diagnostics
\citep{Sekulovski2024}.

\paragraph*{Existing bottleneck: integer-moment checks amplify tail sensitivity, and existing practice requires choosing a ``true'' (or ``more complex'') model.}
The original Good check emphasizes integer-moment instances---most notably the $t=1$ (Turing)
and $t=2$ cases. The $t=2$ check effectively brings in a \emph{higher-order moment}
(equivalently, a squared Bayes factor term) and therefore magnifies tail sensitivity: in principle,
the larger the moment order, the more severely rare extremes influence the Monte Carlo estimate.
This creates two practical complications emphasized by \citet{Sekulovski2024}.

First, numerical stability becomes fragile in exactly the settings where diagnostics are
most needed (weak overlap between $p_1$ and $p_2$, large sample size, sharp marginal likelihoods).
Second, to mitigate this instability, one is often advised to generate synthetic data from the
``more complex'' model when computing the tail-sensitive terms.
However, in many realistic comparisons it is not obvious which model should be regarded as
more complex (e.g., multiple nonlinear models of similar dimension and flexibility), and the
diagnostic becomes contingent on a subjective design choice.

\paragraph*{Design question and resolution: choosing an exponent with two-sided, distribution-free stability.}
The Turing--Good identity holds for a continuum of exponents.
However, Monte Carlo reliability is governed by second moments, and the variance identities
in Lemma~\ref{lem:var-identities} show that the two side-specific risks depend on the shifted overlaps
$I(2t)$ and $I(2t-1)$.
Section~\ref{sec:TG-variance-halforder} proves that $t=\tfrac12$ is the \emph{unique} exponent that
(i) equalizes the two variances and (ii) guarantees a finite worst-side variance uniformly over all
mutually absolutely continuous pairs $(p_1,p_2)$ (Theorem~\ref{thm:minimax}).
The next subsection turns this minimax principle into a concrete two-sided diagnostic and compares it,
at matched Monte Carlo cost, to the one-sided $t=1$ check emphasized by \citet{Sekulovski2024}.

\subsection{A two-sided check and a cost-matched comparison}

Theorem~\ref{thm:TG} implies that, for every $t\in D$,
\[
\E_{\mathcal{H}_2}[B^t]-\E_{\mathcal{H}_1}[B^{t-1}]=0.
\]
A natural way to empirically assess this identity is to estimate the two expectations separately from
independent samples and to compare the corresponding sample means.
The next lemma shows that the half-order choice is again minimax---now for the variance of the two-sided difference.

\begin{lemma}[Unique minimax for the two-sided check]
\label{lem:two-sided}
Let $X_1\sim\mathcal{H}_1$ and $X_2\sim\mathcal{H}_2$ be independent and define
$\Delta_t:=B(X_2)^t-B(X_1)^{t-1}$.
Then, at $t=\half$,
\[
\Var(\Delta_{1/2}) = 2(1-\rho^2)<\infty,
\]
whereas for every $t\neq \half$,
\[
\sup_{(p_1,p_2)\in\mathcal P_{\mathrm{ac}}}\Var(\Delta_t)=+\infty.
\]
\end{lemma}

\begin{proof}
By independence and Lemma~\ref{lem:var-identities},
\[
\Var(\Delta_t)=\Var_{\mathcal{H}_2}(B^t)+\Var_{\mathcal{H}_1}(B^{t-1})
=\{I(2t)-I(t)^2\}+\{I(2t-1)-I(t)^2\}.
\]
At $t=\half$ this becomes $\{I(1)-I(\half)^2\}+\{I(0)-I(\half)^2\}=2(1-\rho^2)$.
For $t\neq \half$, Theorem~\ref{thm:minimax}(ii) provides pairs $(p_1,p_2)$ for which either
$\Var_{\mathcal{H}_2}(B^t)=\infty$ (when $t>\half$) or $\Var_{\mathcal{H}_1}(B^{t-1})=\infty$ (when $t<\half$),
forcing $\Var(\Delta_t)=\infty$ for those pairs.
\end{proof}

Finally, we compare the (balanced) two-sample half-order check to the one-sample Turing check advocated by
\citet{Sekulovski2024}, under a fixed Monte Carlo budget.

We measure Monte Carlo cost by the number $N$ of Bayes factor evaluations (equivalently, draws at which $B=p_1/p_2$ is computed).
Consider:
\begin{enumerate}
\item[(a)] \textit{One-sample Turing check ($t=1$).}
Draw $X_1,\dots,X_N\stackrel{\mathrm{iid}}{\sim}\mathcal{H}_2$ and compute $\bar B_N=N^{-1}\sum_{i=1}^N B(X_i)$.
\item[(b)] \textit{Two-sample half-order check ($t=\half$).}
Draw $X_{11},\dots,X_{1n_1}\stackrel{\mathrm{iid}}{\sim}\mathcal{H}_1$ and
$X_{21},\dots,X_{2n_2}\stackrel{\mathrm{iid}}{\sim}\mathcal{H}_2$ independently, with $n_1+n_2=N$, and compute
$\bar B^{-1/2}_{1n_1}=n_1^{-1}\sum_{i=1}^{n_1}B(X_{1i})^{-1/2}$ and
$\bar B^{1/2}_{2n_2}=n_2^{-1}\sum_{i=1}^{n_2}B(X_{2i})^{1/2}$.
\end{enumerate}
For the one-sample check, the natural target is $\bar B_N-1$.
For the two-sample check, the natural target is $\bar B^{1/2}_{2n_2}-\bar B^{-1/2}_{1n_1}$.

\medskip
\noindent\textbf{Convention on the balanced split.}
When we write the balanced allocation as $n_1=n_2=N/2$, we tacitly assume that $N$ is even.
If $N$ is odd, one may take the nearest-integer split
$(n_1,n_2)=(\lfloor N/2\rfloor,\lceil N/2\rceil)$; this changes the displayed constants and
thresholds only by a factor $1+O(N^{-2})$.

\begin{theorem}[Worst-case variance bound and conditional matched-cost dominance]\label{thm:matched-cost}
\label{thm:unified}
Let $V_{(1)}:=\Var(\bar B_N-1)$ and let $V_{(2)}(n_1,n_2):=\Var(\bar B^{1/2}_{2n_2}-\bar B^{-1/2}_{1n_1})$, with $n_1+n_2=N$.
For the balanced split, define $V_{(2)}^{\mathrm{bal}}:=V_{(2)}(N/2,N/2)$.
Then:
\begin{enumerate}
\item[(i)] \textnormal{(Worst-case boundedness.)}
For every $(p_1,p_2)\in\mathcal P_{\mathrm{ac}}$,
under the balanced two-sample choice $n_1=n_2=N/2$ (with $N$ even),
\[
V_{(2)}^{\mathrm{bal}}=V_{(2)}(N/2,N/2)=\frac{4(1-\rho^2)}{N},
\]
hence $\sup_{(p_1,p_2)\in\mathcal P_{\mathrm{ac}}}V_{(2)}^{\mathrm{bal}}\le 4/N<\infty$.
\item[(ii)] \textnormal{(Conditional dominance.)}
Under the balanced split (with $N$ even),
whenever $\rho\le 1/2$,
$V_{(2)}^{\mathrm{bal}}\le V_{(1)}$.
\end{enumerate}
\end{theorem}

\begin{proof}
For the one-sample check, independence gives
\[
V_{(1)}=\Var_{\mathcal{H}_2}(B)/N=\{I(2)-I(1)^2\}/N=(I(2)-1)/N,
\]
where we used Lemma~\ref{lem:var-identities} at $t=1$ and Lemma~\ref{lem:basic-bounds}(i).
For the two-sample half-order check, independence and Corollary~\ref{cor:half-order-eq} yield
\[
V_{(2)}(n_1,n_2)
=\frac{1}{n_2}\Var_{\mathcal{H}_2}(B^{1/2})+\frac{1}{n_1}\Var_{\mathcal{H}_1}(B^{-1/2})
=\Bigl(\frac{1}{n_1}+\frac{1}{n_2}\Bigr)(1-\rho^2).
\]
For fixed $n_1+n_2=N$, the AM--GM inequality implies that $n_1n_2$ is maximized at $n_1=n_2=N/2$,
hence the balanced choice minimizes $V_{(2)}(n_1,n_2)$ and gives $V_{(2)}=4(1-\rho^2)/N$.

\smallskip\noindent
(i) Since $\rho\in(0,1]$, we have $1-\rho^2\le 1$, so $\sup_{(p_1,p_2)}V_{(2)}\le 4/N$.
On the other hand, by Theorem~\ref{thm:minimax}(ii), there exist pairs for which $I(2)=\infty$,
which implies $V_{(1)}=(I(2)-1)/N=\infty$.

\smallskip\noindent
(ii) By~\eqref{eq:ratio-right} at $t=1$,
$I(2)\ge\rho^{-2}$ (since $I(1)=1$), hence
\[
V(1)=\frac{I(2)-1}{N}\ge \frac{\rho^{-2}-1}{N}=\frac{1-\rho^2}{N\rho^2}.
\]
Under the balanced split $n_1=n_2=N/2$, we have $V_{(2)}^{\mathrm{bal}}=4\frac{1-\rho^2}{N}$.
Therefore $V(2)\le V(1)$ whenever
\[
\frac{4(1-\rho^2)}{N}\le \frac{1-\rho^2}{N\rho^2},
\]
i.e., whenever $4\le \rho^{-2}$, equivalently $\rho\le 1/2$.
\end{proof}

\paragraph*{Remark (When $\rho>\half$).}
When the two models are very similar (high overlap), $\rho$ can exceed $1/2$.
In this regime, the ordering between $V_{(1)}$ and $V_{(2)}$ is not determined by $\rho$ alone:
under the balanced split,
\[
V_{(2)}\le V_{(1)}
\quad\Longleftrightarrow\quad
I(2)\ \ge\ 1+4(1-\rho^2).
\]
Since $I(2)=\int p_1(x)^2/p_2(x)\,dx=1+\chi^2(p_1\|p_2)$, this condition is equivalent to
$\chi^2(p_1\|p_2)\ge 4(1-\rho^2)$.
Convexity alone guarantees only $I(2)\ge \rho^{-2}$, so either ordering is possible when $\rho>\half$.
Nevertheless, the worst-case comparison remains favorable to the half-order two-sample check:
$V_{(2)}$ is always finite and uniformly bounded by $4/N$, whereas $V_{(1)}$ can be arbitrarily large.

\subsection{Practical implementation and simulation studies}\label{subsec:goodcheck-sim}

The half-order theory suggests a simple default: estimate the overlap $\rho=I(\tfrac12)$ from \emph{both} generating models
and check agreement.
This yields a symmetric diagnostic that does not require deciding which model is ``true'' or ``more complex'' and, by
Theorem~\ref{thm:minimax}, guarantees bounded per-draw variance on both sides.
We summarize the resulting workflow and then report simulation studies.

\paragraph*{A half-order Good check.}
Let $B(x)=p_1(x)/p_2(x)$ be the Bayes factor computed by the numerical procedure under scrutiny.
A half-order Good check compares the two sides of Theorem~\ref{thm:TG} at $t=\tfrac12$:
\[
\mathbb{E}_{\mathcal{H}_2}\!\left[B^{1/2}\right]
\;=\;
\mathbb{E}_{\mathcal{H}_1}\!\left[B^{-1/2}\right]
\;=\;\rho.
\]
Algorithm~\ref{alg:two-sided-halforder} summarizes a one-box recipe for implementing this two-sided diagnostic.

\refstepcounter{algobox}\label{alg:two-sided-halforder}
\begin{center}
\fbox{\begin{minipage}{0.95\linewidth}
\textbf{Algorithm \thealgobox: Two-sided half-order Good check.}

\smallskip
\noindent\textbf{Input:} prior-predictive simulators under $\mathcal{H}_1$ and $\mathcal{H}_2$; a Bayes factor evaluator $B(x)=p_1(x)/p_2(x)$.\\
\textbf{Output:} $(\Delta,\widehat{\mathrm{se}}(\Delta),\widehat\rho)$.

\begin{enumerate}[label=\textbf{Step \arabic*:},leftmargin=*]
\item Choose Monte Carlo budgets $m_1,m_2$ (default $m_1=m_2$) and a tolerance level $\varepsilon$.
\item Draw $X_{2i}\stackrel{\text{i.i.d.}}{\sim} \mathcal{H}_2$ and compute $U_i:=B(X_{2i})^{1/2}$ for $i=1,\dots,m_2$.
\item Draw $X_{1i}\stackrel{\text{i.i.d.}}{\sim} \mathcal{H}_1$ and compute $V_i:=B(X_{1i})^{-1/2}$ for $i=1,\dots,m_1$.
\item Compute
\[
\widehat\rho_{\,2}:=\frac{1}{m_2}\sum_{i=1}^{m_2}U_i,\qquad
\widehat\rho_{\,1}:=\frac{1}{m_1}\sum_{i=1}^{m_1}V_i,\qquad
\Delta:=\widehat\rho_{\,2}-\widehat\rho_{\,1},\qquad
\widehat\rho:=\frac{\widehat\rho_{\,1}+\widehat\rho_{\,2}}{2}.
\]
\item Check if $\Delta$ is close enough to zero. If one wants an objective thteshold, one option may be to estimate $\widehat{\mathrm{se}}(\Delta)$ and calibrate $\varepsilon$ using the rules in Section~3.5; persistent $|\Delta|>\varepsilon$ flags a simulator--evaluator mismatch.
\end{enumerate}
\end{minipage}}
\end{center}

Under correct Bayes factor computation, both estimators target the \emph{same} quantity $\rho$, so the discrepancy $\Delta$
should be close to $0$ up to Monte Carlo error.
Because each summand has bounded variance at $t=\tfrac12$ (Theorem~\ref{thm:minimax}),
this check is intrinsically less sensitive to rare, extreme Bayes factor values than the integer-moment
choices emphasized in the original Good check.


\paragraph*{Symmetry: no need to decide which model is ``true'' (or ``more complex'').}
A further practical advantage of the half-order choice is symmetry between the two model orientations.
Writing $B=p_1/p_2$ (so $B^{-1}=p_2/p_1$), the half-order identity is
\[
\E_{\Htwo}\!\left[B^{1/2}\right]
=
\E_{\Hone}\!\left[B^{-1/2}\right]
=\rho.
\]
Equivalently, one may swap the model labels (i.e., replace $B$ by $B^{-1}$) without changing the diagnostic.
Thus the check need not privilege one model as ``true'' or ``more complex'': whichever direction is easier to
simulate or implement, the same theoretical target and the same variance guarantee apply.

\paragraph*{Matched-budget comparison: when overlap is small, the half-order two-sided check can dominate the one-sided Turing check.}
A common practical constraint is a fixed budget of $N$ Bayes factor evaluations.
Under such a budget, Sekulovski et al.'s one-sided Turing check targets $t=1$ with variance
proportional to $I(2)-1=\mathbb{E}_{\mathcal{H}_2}[B^2]-1$, which may be arbitrarily large or even
infinite.
By contrast, the balanced two-sided half-order design has variance proportional to
$1-\rho^2\le 1$ and is uniformly bounded in the worst case; see
Theorem~\ref{thm:unified} for a formal matched-cost comparison and dominance conditions.
Hence, even when one is primarily interested in checking Corollary~\ref{cor:Turing},
the half-order two-sided check provides a conservative default with minimax-stable behavior.



\subsection{Simulation studies}\label{subsec:sim-studies}

We use a binomial point-null example (adapted from \citet{Sekulovski2024}) to compare finite-run Monte Carlo behavior
across three Good-check choices, $t=\half,1,2$, and to illustrate sensitivity to simulator--evaluator mismatch.
Throughout, $B(y):=p_1(y)/p_2(y)$ denotes the Bayes factor in favor of $\Hone$ over $\Htwo$.
For each configuration and each $n\in\{10,50,100\}$, we conducted $R=10{,}000$ independent Good-check runs.
Each run uses $m=2{,}000$ prior-predictive draws from $\Hone$ and $m=2{,}000$ prior-predictive draws from $\Htwo$,
and we record Monte Carlo \emph{differences} for the half-order identity and for the integer-order ($t=1,2$) identities.
Tables report Mean (SD) over the $R$ runs.
Figures visualize the finite-$m$ behavior for $n=50$ using fan charts (mean and central 50\%/90\% bands across runs).

\subsubsection{Simulation study 1A: Binomial example (correct specification)}\label{subsec:sim1A_binom}

Model:
\[
  Y\mid \theta \sim \mathrm{Binomial}(n,\theta),\qquad
  \Htwo:\theta=\half,\qquad
  \Hone:\theta\sim \mathrm{Beta}(1,1).
\]
Bayes factors are evaluated under the same (correct) specification.
Table~\ref{tab:sim1_binom_correct} summarizes Monte Carlo differences for the half-order discrepancy
$\widehat\Delta_{1/2}:=\widehat{\E}_{\Htwo}[B^{1/2}]-\widehat{\E}_{\Hone}[B^{-1/2}]$ and for the $t=1,2$ checks
(both forward and reverse orientations).

\input{figs/Table1_1A_rev2.tex}

Table~\ref{tab:sim1_binom_correct} shows that the half-order two-sided discrepancy remains tightly centered near zero with
small run-to-run variability across all $n$.
In contrast, the \emph{forward} integer-order quantities involving $\widehat{\E}_{\Htwo}[B]$ and especially
$\widehat{\E}_{\Htwo}[B^2]$ can be extremely unstable in finite Monte Carlo runs because they are dominated by exceedingly
rare outcomes under $\Htwo$.
The reverse-oriented integer checks are markedly more stable on this scale.

Figure~\ref{fig:sim1_outliers_sqrt} complements the table by displaying the run-to-run distribution of \emph{running}
discrepancies as a function of $m$ for $n=50$ (half-order discrepancy, reverse Turing, and reverse Good).

\begin{figure}[!htbp]
  \centering
  \includegraphics[width=0.86\linewidth]{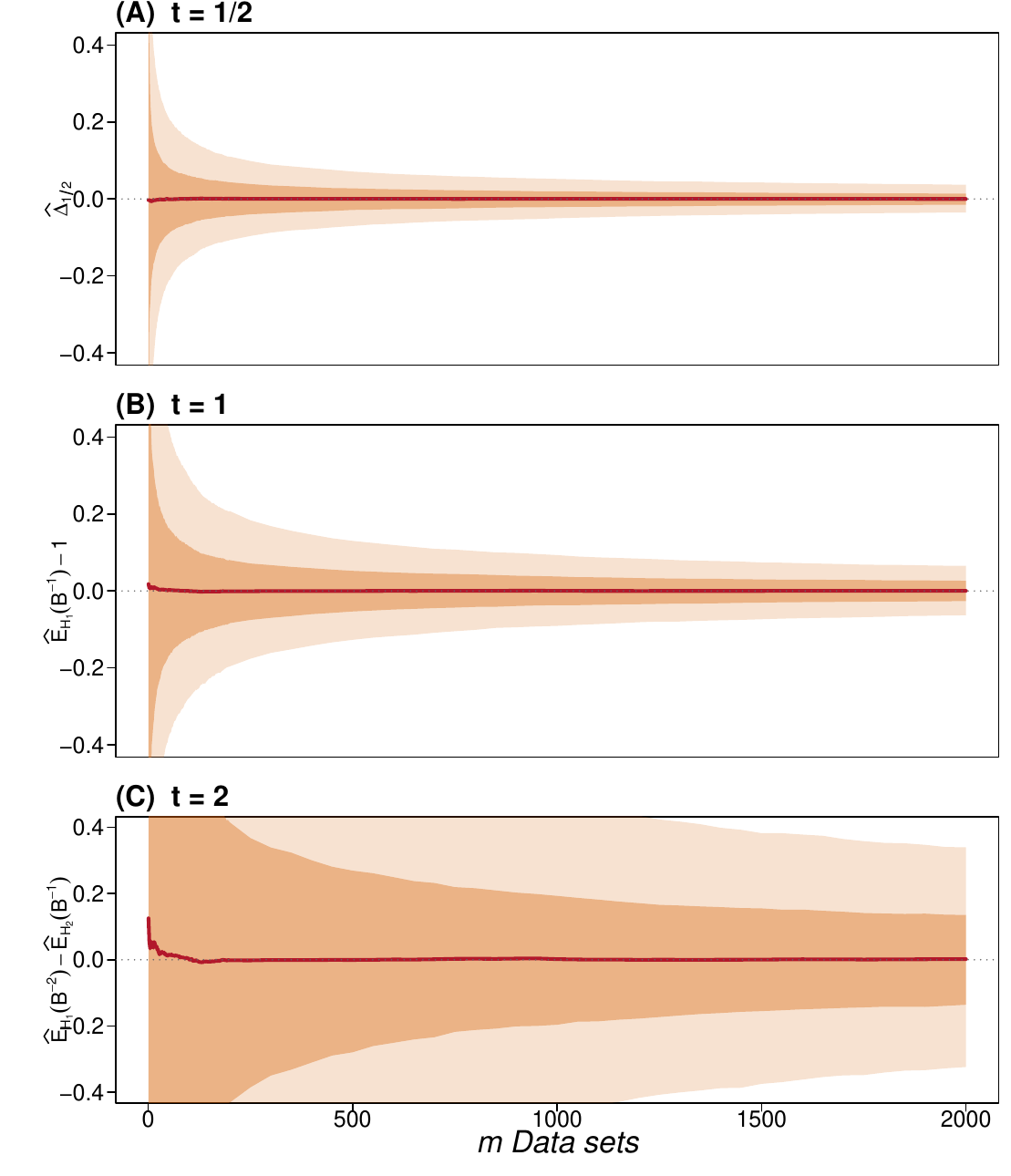}
  \caption{Simulation study~1A (binomial example, correctly specified; $n=50$): fan charts of running Monte Carlo
  discrepancies as a function of $m$ (number of simulated datasets per hypothesis) across $R=10{,}000$ independent runs.
  Panel (A) shows $\widehat\Delta_{1/2}(m)=\widehat{\E}_{\Htwo,\le m}[B^{1/2}]-\widehat{\E}_{\Hone,\le m}[B^{-1/2}]$.
  Panel (B) shows the reverse Turing discrepancy $\widehat{\E}_{\Hone,\le m}[B^{-1}]-1$.
  Panel (C) shows the reverse Good discrepancy $\widehat{\E}_{\Hone,\le m}[B^{-2}]-\widehat{\E}_{\Htwo,\le m}[B^{-1}]$.
  Solid curves are means; shaded bands are central 50\% (dark) and 90\% (light) intervals; the horizontal dotted line
  indicates $0$ (the exact identity).}
  \label{fig:sim1_outliers_sqrt}
\end{figure}

\FloatBarrier

\subsubsection{Simulation study 1B: Binomial example (simulator--evaluator mismatch)}\label{subsec:sim1B_binom}

To emulate a mild misspecification, we keep the \emph{Bayes factor evaluator} fixed at the intended model
$\Hone:\theta\sim\mathrm{Beta}(1,1)$, but generate synthetic data under $\Hone$ using a slightly different prior:
\[
  Y\mid \theta \sim \mathrm{Binomial}(n,\theta),\qquad
  \Htwo:\theta=\half,\qquad
  \text{simulator under }\Hone:\ \theta\sim \mathrm{Beta}(1.2,1.2).
\]
This simulator--evaluator mismatch violates the exact Turing--Good identities, and should therefore induce systematic,
persistent discrepancies.

Table~\ref{tab:sim1_binom_mismatch} reports Mean (SD) differences across $R=10{,}000$ runs.
The half-order discrepancy $\widehat\Delta_{1/2}$ is shifted away from zero with comparatively small variability,
indicating clear detectability at moderate Monte Carlo budgets.
As in the correctly specified case, forward integer-order quantities under $\Htwo$ can be highly unstable.

\input{figs/Table1_1B_rev2.tex}

Figure~\ref{fig:sim1_mismatch_discrepancy} visualizes the same phenomenon for $n=50$ via fan charts of running
discrepancies. The vertical blue line marks the first $m$ (on the plotting grid) at which the central 90\% band no longer
contains $0$, providing a simple visual ``detection time'' summary.

\begin{figure}[!htbp]
  \centering
  \includegraphics[width=0.86\linewidth]{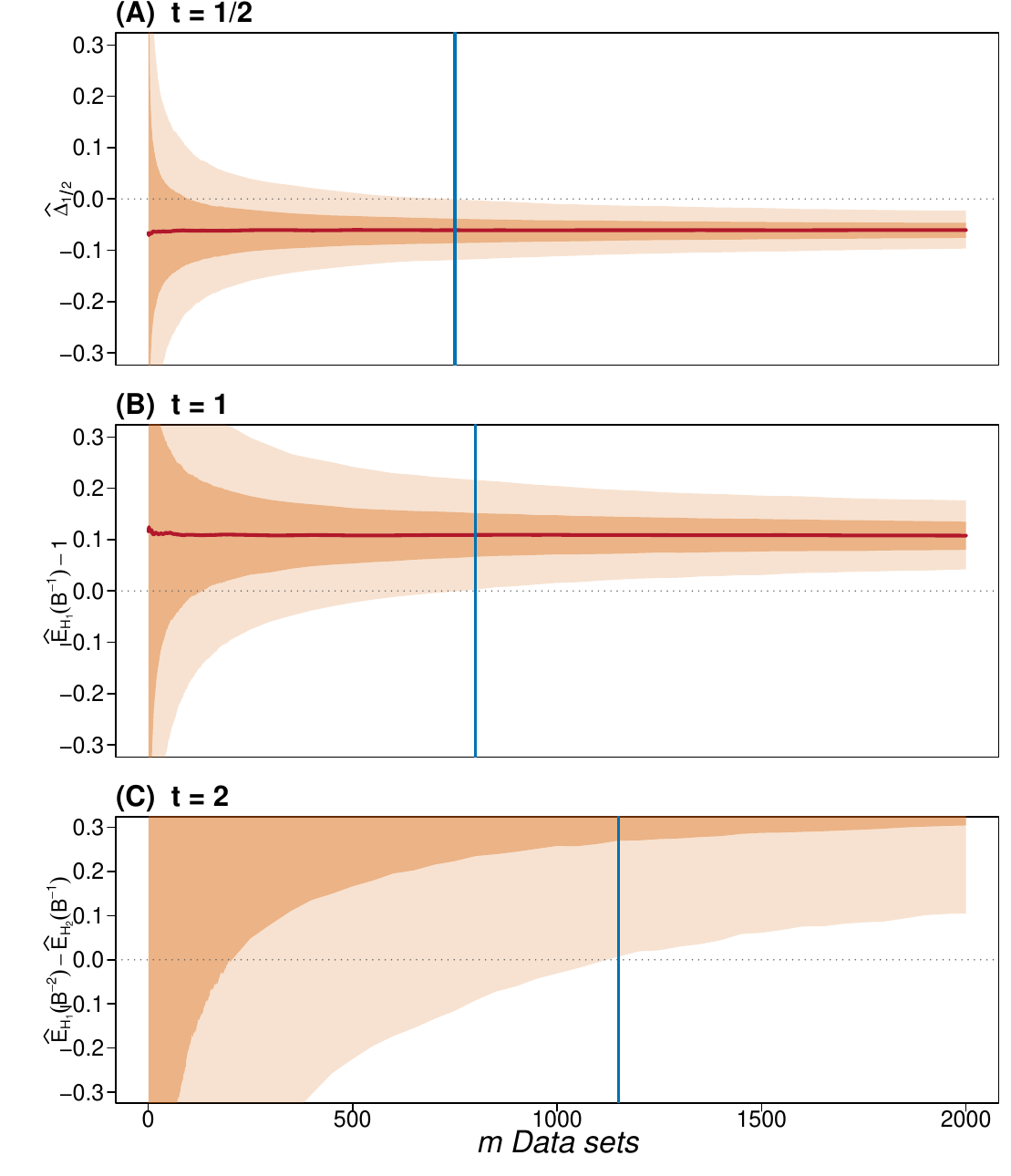}
  \caption{Simulation study~1B (binomial example; simulator--evaluator mismatch; $n=50$): fan charts of running Monte Carlo
  discrepancies across $R=10{,}000$ runs. Panels (A)--(C) match Figure~\ref{fig:sim1_outliers_sqrt}.
  The vertical blue line indicates the first $m$ at which the central 90\% band excludes $0$ (on the plotted $m$ grid).}
  \label{fig:sim1_mismatch_discrepancy}
\end{figure}

\FloatBarrier
\paragraph*{Takeaway.}
Across both workflows, the empirical behavior mirrors the theoretical message of Section~\ref{sec:TG-variance-halforder}:
Monte Carlo reliability of Good checks is controlled by shifted overlaps such as $I(2t)$ and can deteriorate rapidly away from
half order, often in a strongly asymmetric way across the two orientations.
The two-sided half-order check remains stable across model separations (as reflected by $\rho$) and retains sensitivity to
simulator--evaluator mismatches, including mild prior perturbations, without requiring a judgement about which model should be
treated as ``true'' or ``more complex.''

\subsection{Implications for practice and design}\label{sec:implications}


\paragraph*{(1) Default recommendation: a balanced two-sided half-order check.}
For Monte Carlo verification of the Turing--Good identity (Theorem~\ref{thm:TG}),
the half-order exponent $t=\half$ is uniquely stable in second moments:
it equalizes the two side-specific variances at $1-\rho^2$ (Corollary~\ref{cor:half-order-eq})
and uniquely minimizes the worst-side variance risk $R(t)$ (Theorem~\ref{thm:minimax}).
Operationally, the natural \emph{two-sided} estimator compares
$\bar B^{1/2}_{2n_2}$ and $\bar B^{-1/2}_{1n_1}$ using independent draws from $\mathcal{H}_2$ and $\mathcal{H}_1$,
respectively.
For a fixed total budget $N=n_1+n_2$ Bayes factor evaluations, Theorem~\ref{thm:unified} shows that the variance
$V_{(2)}(n_1,n_2)=\big(\frac1{n_1}+\frac1{n_2}\big)(1-\rho^2)$ is minimized by the balanced split
$n_1=n_2=N/2$ (see the convention on the balanced split preceding Theorem~\ref{thm:matched-cost}),
giving $V(2)=4(1-\rho^2)/N$. In particular, $\sup_{(p_1,p_2)\in\mathcal P_{\mathrm{ac}}}V_{(2)}^{\mathrm{bal}}\le 4/N$.
%

 In particular, $\sup V(2)\le 4/N$.

\paragraph*{(2) Small-overlap regime ($\rho\le\half$): two-sided half-order dominates the one-sided Turing check.}
Under the same Monte Carlo budget $N$, the one-sided Turing check based on $\bar B_N$ has variance
$V_{(1)}=(I(2)-1)/N$ (Theorem~\ref{thm:unified}).
Theorem~\ref{thm:unified}(ii) shows that, under the balanced split, whenever $\rho\le\half$,
\[
V_{(2)} \le V_{(1)}.
\]
Equivalently, the relative efficiency
\[
\mathrm{RE}:=\frac{V_{(1)}}{V_{(2)}}=\frac{I(2)-1}{4(1-\rho^2)}
\]
satisfies the universal lower bound $\mathrm{RE}\ge 1/(4\rho^2)$ (using $I(2)\ge \rho^{-2}$ from~\eqref{eq:ratio-right} at $t=1$).
Thus, at $\rho=\half$ the procedures are (at least) equally efficient, while at $\rho=0.25$ the half-order two-sided check is guaranteed to be $\ge 4\times$ as efficient at matched cost.

\paragraph*{(3) Reporting and interpreting the overlap $\rho$.}
The overlap $\rho=\int\sqrt{p_1p_2}\,dx$ is the Bhattacharyya coefficient (Hellinger affinity) between the two marginal likelihoods.
It is the \emph{only} model-dependent quantity governing the bounded two-sided half-order variance $V_{(2)}$.
For communication and planning, it is therefore useful to report either $\rho$ itself or $1-\rho^2$
(which equals the common half-order variance under either model; Corollary~\ref{cor:half-order-eq}).
Moreover, once $\bar B^{1/2}_{2n_2}$ and $\bar B^{-1/2}_{1n_1}$ agree within Monte Carlo error, either can be reported as an estimate of~$\rho$.

\paragraph*{(4) Cost-aware allocation and conservative thresholds.}
If the unit costs of generating a draw and evaluating $B$ differ under $\mathcal{H}_1$ and $\mathcal{H}_2$
(say $c_1$ and $c_2$), one can minimize $V_{(2)}(n_1,n_2)$ under a cost constraint $c_1n_1+c_2n_2=C$.
A Lagrange-multiplier calculation yields the optimal ratio
\[
n_1:n_2 \ \propto\ \frac{1}{\sqrt{c_1}}:\frac{1}{\sqrt{c_2}}.
\]
When $\rho$ is unknown, the worst-case bound $V_{(2)}\le 4/N$ (balanced split) implies a simple distribution-free Chebyshev guarantee:
\[
\mathbb{P}\!\left(\bigl|\bar B^{1/2}_{2n_2}-\bar B^{-1/2}_{1n_1}\bigr|\ge \varepsilon\right)
\le \frac{V_{(2)}}{\varepsilon^2}\le \frac{4}{N\varepsilon^2}.
\]
This can be used to choose conservative tolerance levels for the diagnostic before any pilot estimate of $\rho$ is available.

\paragraph*{(5) Standard error}
Let
\[
\Delta \;:=\; \widehat\rho_{\,2}-\widehat\rho_{\,1}
\]
be the observed discrepancy in the two-sided half-order check (Algorithm~\ref{alg:two-sided-halforder}; Section~3.3).
Under correct Bayes factor computation, $\E[\Delta]=0$ and, by Theorem~\ref{thm:unified},
\[
\Var(\Delta)=(1-\rho^2)\Bigl(\frac{1}{m_1}+\frac{1}{m_2}\Bigr).
\]
In practice, 

one can compute the sample variances
\[
s_2^2:=\frac{1}{m_2-1}\sum_{i=1}^{m_2}\Bigl\{B(X_{2i})^{1/2}-\widehat\rho_{\,2}\Bigr\}^2,
\qquad
s_1^2:=\frac{1}{m_1-1}\sum_{i=1}^{m_1}\Bigl\{B(X_{1i})^{-1/2}-\widehat\rho_{\,1}\Bigr\}^2,
\]
and the estimated standard error
\[
\widehat{\mathrm{se}}(\Delta):=\sqrt{\frac{s_2^2}{m_2}+\frac{s_1^2}{m_1}}.
\]
Then the studentized statistic $T:=\Delta/\widehat{\mathrm{se}}(\Delta)$ is approximately standard normal for moderate
$(m_1,m_2)$, suggesting the tolerance
\[
\varepsilon_{\mathrm{CLT}}(\alpha):=z_{1-\alpha/2}\,\widehat{\mathrm{se}}(\Delta),
\]
where $z_{1-\alpha/2}$ is the $(1-\alpha/2)$ quantile of the standard normal.

A conservative alternative that uses only the distribution-free bound
$1-\rho^2\le 1$ is to bound the \emph{true} standard error:
\[
  \mathrm{se}(\Delta)
  := \sqrt{\frac{\Var_{\mathcal{H}_2}(B^{1/2})}{m_2}+\frac{\Var_{\mathcal{H}_1}(B^{-1/2})}{m_1}}
  = \sqrt{1-\rho^2}\,\sqrt{\frac1{m_1}+\frac1{m_2}}
  \le \sqrt{\frac1{m_1}+\frac1{m_2}}.
\]
Define $\mathrm{se}_{\mathrm{wc}}(\Delta):=\sqrt{1/m_1+1/m_2}$ and use
$\varepsilon^{\mathrm{wc}}_{\mathrm{CLT}}(\alpha):=z_{1-\alpha/2}\,\mathrm{se}_{\mathrm{wc}}(\Delta)$.

\paragraph*{(6) When a one-sided check is unavoidable.}
If simulation is feasible from only one model, then a two-sided check is impossible and one must rely on a one-sided identity:
$t=1$ under $\mathcal{H}_2$ (Corollary~\ref{cor:Turing}) or, symmetrically, $t=0$ under $\mathcal{H}_1$
(since $I(0)=1$ and $\E_{\mathcal{H}_1}[B^{-1}]=1$ by Theorem~\ref{thm:TG}).
Theorem~\ref{thm:minimax}(ii) cautions that such one-sided checks can have unbounded variance for some mutually absolutely continuous pairs,
so in practice they should be accompanied by tail diagnostics for $B$ (or $\log B$), sufficiently large Monte Carlo budgets, and---whenever feasible---a return to the two-sided half-order design.

\section{Broader implications of the half-order exponent}\label{sec:broader-implications}

Section~\ref{sec:TG-variance-halforder} identifies the half-order exponent $t=\tfrac12$
as the \emph{unique} choice
that equalizes and uniformly bounds the Monte Carlo variances on the $\mathcal{H}_1$ and $\mathcal{H}_2$
sides (Theorem~\ref{thm:minimax}).  Although our motivating application was \emph{Good checks}
for Bayes factor implementations \citep{Sekulovski2024}, the underlying structure is more general:
many core problems in Bayesian computation reduce to estimating expectations of powers of a density ratio,
and their numerical stability is governed by the same overlap family $I(t)=\int p_1^t p_2^{1-t}\,d\mu$.

We begin with the most central computational instance: \emph{ratios of normalizing constants}.
This framing subsumes marginal likelihood ratios, Bayes factors, and (more broadly) free-energy differences.
We show that the half-order midpoint $\sqrt{p_1p_2}$ yields a canonical ``geometric bridge'' for ratio estimation,
together with a stably estimable overlap $\rho=I(\tfrac12)$ that directly controls the relative Monte Carlo error.
We then connect the same half-order structure to further settings (testing-theoretic bounds, distribution shift,
and robust evidence measures) in the subsequent subsections.

\subsection{Normalizing-constant ratios: the geometric bridge and overlap-based design}\label{subsec:normconst}

\paragraph*{Setup: ratios of normalizing constants as the computational core.}
Let $\tilde p_1$ and $\tilde p_2$ be nonnegative integrable functions on $(\mathcal X,\mathcal A,\mu)$ with
unknown normalizing constants $Z_j=\int \tilde p_j\,d\mu\in(0,\infty)$ and normalized densities
$p_j=\tilde p_j/Z_j$.  The fundamental target is the ratio
\[
r:=\frac{Z_1}{Z_2},
\]
which includes Bayes factors as a special case (e.g., $Z_j$ as marginal likelihood under model $j$).
In practice, $Z_j$ is typically intractable even when simulation from $p_j$ is feasible (e.g., by MCMC),
so computation hinges on identities that express $r$ as an expectation under one or both of $p_1,p_2$.

\paragraph*{Moment fragility of one-sided ratio estimators.}
Define the unnormalized ratio
\[
w(x):=\frac{\tilde p_1(x)}{\tilde p_2(x)}.
\]
A baseline identity is the one-sided importance-sampling representation
\[
\frac{Z_1}{Z_2}=\E_{p_2}\!\left[w(X)\right],
\]
so $\widehat r_{\mathrm{IS}}=\frac1n\sum_{i=1}^n w(X_i)$ with $X_i\sim p_2$ is unbiased.
However, its reliability is governed by $\Var_{p_2}(w)$ (equivalently $\E_{p_2}[w^2]$),
which can be enormous or infinite under mild tail mismatch or limited overlap.
This is a core reason why bridge sampling and related stabilized estimators are widely used
\citep{MengWong1996,GronauJMP2017}.

\paragraph*{A robust design desideratum.}
General bridge-sampling identities introduce a bridge function $h$ and use samples from both $p_1$ and $p_2$
to stabilize ratio estimation \citep{MengWong1996,GronauJMP2017}; related optimality ideas also appear
in free-energy computation \citep{Bennett1976,ShirtsChodera2008}.
Yet in the regimes that motivate bridge/path sampling (small overlap; heavy tails; high dimension),
many candidate transformations still deteriorate sharply, and the choice of tuning criteria
based on higher moments of density ratios can become ill-posed.
This motivates two practical questions:
\begin{enumerate}[(i)]
\item which ratio transformation provides a \emph{distribution-free} stability baseline?
\item which overlap quantity can be estimated reliably enough to guide reference choice and bridging 
design?
\end{enumerate}

\paragraph*{Half order as the uniquely worst-case-stable midpoint.}
Section~\ref{sec:TG-variance-halforder} answers (i)--(ii) in a unified manner.
Applied to the normalized density ratio $B=p_1/p_2$, the Turing--Good identity (Theorem~\ref{thm:TG}) gives
$\E_{\mathcal{H}_2}[B^t]=\E_{\mathcal{H}_1}[B^{t-1}]=I(t)$ whenever these expectations are finite.
Crucially, Theorem~\ref{thm:minimax} shows that $t=\tfrac12$ is the \emph{unique} exponent that
equalizes and uniformly bounds the worst-side variance across all mutually absolutely continuous pairs.
Translating back to normalizing-constant ratios yields a canonical ``geometric bridge''.

\begin{proposition}[Half-order bridge identity for normalizing-constant ratios]
\label{prop:halforder-bridge}
In the above setting, assume $p_1$ and $p_2$ are mutually absolutely continuous.
Define $w(x):=\tilde p_1(x)/\tilde p_2(x)$, $r:=Z_1/Z_2$,
and the half-order overlap
\[
\rho:=\int \sqrt{p_1(x)p_2(x)}\,d\mu \in (0,1].
\]
Then
\[
\E_{p_2}\!\left[w(X)^{1/2}\right]=r^{1/2}\rho,
\qquad
\E_{p_1}\!\left[w(X)^{-1/2}\right]=r^{-1/2}\rho,
\]
and hence the ratio admits the \emph{half-order bridge identity}
\begin{equation}\label{eq:halforder-bridge}
\frac{Z_1}{Z_2}
=\frac{\E_{p_2}\!\left[w(X)^{1/2}\right]}{\E_{p_1}\!\left[w(X)^{-1/2}\right]} .
\end{equation}
\end{proposition}

\begin{proof}
By definition,
$\E_{p_2}[w^{1/2}]
=Z_2^{-1}\int \tilde p_2(\tilde p_1/\tilde p_2)^{1/2}\,d\mu
=Z_2^{-1}\int \sqrt{\tilde p_1\tilde p_2}\,d\mu$.
Since $\sqrt{\tilde p_1\tilde p_2}=\sqrt{Z_1Z_2}\,\sqrt{p_1p_2}$, this equals
$\sqrt{Z_1/Z_2}\int\sqrt{p_1p_2}\,d\mu=r^{1/2}\rho$.
The $p_1$ identity is analogous, and taking the ratio yields~\eqref{eq:halforder-bridge}.
\end{proof}

\begin{proposition}[Second-moment stability and an overlap-controlled variance formula]
\label{prop:halforder-bridge-stability}
Under the conditions of Proposition~\ref{prop:halforder-bridge},
\[
\Var_{p_2}\!\left(w(X)^{1/2}\right)=r(1-\rho^2),
\qquad
\Var_{p_1}\!\left(w(X)^{-1/2}\right)=r^{-1}(1-\rho^2).
\]
In particular, the half-order building blocks have finite variance for \emph{all} mutually absolutely continuous pairs.
Moreover, their coefficients of variation coincide:
\[
\frac{\Var_{p_2}(w^{1/2})}{\E_{p_2}[w^{1/2}]^2}
=
\frac{\Var_{p_1}(w^{-1/2})}{\E_{p_1}[w^{-1/2}]^2}
=
\rho^{-2}-1.
\]
With independent samples $X_{21},\dots,X_{2m_2}\stackrel{\mathrm{iid}}{\sim}p_2$ and
$X_{11},\dots,X_{1m_1}\stackrel{\mathrm{iid}}{\sim}p_1$, define
\[
\widehat r_{1/2}
:=
\frac{\widehat a_2}{\widehat a_1},
\qquad
\widehat a_2:=\frac{1}{m_2}\sum_{i=1}^{m_2} w(X_{2i})^{1/2},
\quad
\widehat a_1:=\frac{1}{m_1}\sum_{i=1}^{m_1} w(X_{1i})^{-1/2}.
\]
Then, as $m_1,m_2\to\infty$,
\[
\Var(\widehat r_{1/2})
=
r^2\,\frac{1-\rho^2}{\rho^2}\Bigl(\frac{1}{m_1}+\frac{1}{m_2}\Bigr)
+o\!\left(\frac{1}{m_1}+\frac{1}{m_2}\right).
\]
\end{proposition}

\begin{proof}
Write $B=p_1/p_2$ so that $w=rB$. Then $w^{1/2}=r^{1/2}B^{1/2}$ and $w^{-1/2}=r^{-1/2}B^{-1/2}$.
Under $p_2$, $\E[B^{1/2}]=\rho$ and $\Var(B^{1/2})=1-\rho^2$ by Corollary~\ref{cor:half-order-eq},
so $\Var_{p_2}(w^{1/2})=r\,\Var_{p_2}(B^{1/2})=r(1-\rho^2)$; the $p_1$ side is identical.
The coefficient-of-variation identity follows by dividing by the squared means from Proposition~\ref{prop:halforder-bridge}.
Finally, the variance expansion for $\widehat r_{1/2}$ follows from the CLT for $(\widehat a_2,\widehat a_1)$ and
a first-order delta method for $(a_2,a_1)\mapsto a_2/a_1$.
\end{proof}

\paragraph*{Remark (Overlap estimation at essentially no additional cost).}
Proposition~\ref{prop:halforder-bridge} implies the product identity
$\E_{p_2}[w^{1/2}]\,\E_{p_1}[w^{-1/2}]=\rho^2$.
Therefore, if $\widehat a_2$ and $\widehat a_1$ denote the sample means of $w^{1/2}$ under $p_2$ and
$w^{-1/2}$ under $p_1$ as in Proposition~\ref{prop:halforder-bridge-stability}, then
\[
\widehat\rho^2:=\widehat a_2\,\widehat a_1,
\qquad
\widehat\rho:=(\widehat a_2\,\widehat a_1)^{1/2},
\]
is available alongside the ratio estimator $\widehat r_{1/2}=\widehat a_2/\widehat a_1$ at essentially no additional cost.
Because each half-order summand has finite variance for all mutually absolutely continuous pairs,
$\widehat\rho$ provides a stable pilot diagnostic of overlap that can guide reference choice and the refinement of
intermediate bridging
schedules (Section~\ref{subsec:is-degeneracy}).

\paragraph*{Remark (Relation to bridge sampling and what is distinctive here).}
Identity~\eqref{eq:halforder-bridge} is a geometric special case of bridge sampling
\citep{MengWong1996,GronauJMP2017}.
The point emphasized here is not the existence of a bridge identity per se,
but that $t=\tfrac12$ inherits the \emph{distribution-free} worst-case second-moment guarantee from
Section~\ref{sec:TG-variance-halforder}: regardless of tail behavior (as long as mutual absolute continuity holds),
the half-order summands admit finite variance and their relative variability is governed by $\rho$ alone.
This provides a conservative default (and a reliable pilot quantity) even when higher-moment criteria are unstable.

\paragraph*{Design implication for marginal likelihood computation.}
In marginal likelihood estimation one often compares $\tilde p(\theta)=p(y\mid\theta)\pi(\theta)$
to a tractable reference $\tilde g(\theta)$ with known normalizing constant,
reducing the problem to estimating $Z_{\tilde p}/Z_{\tilde g}$ via bridge/path sampling
\citep{MengWong1996,GronauJMP2017,Gronau2017}.
Propositions~\ref{prop:halforder-bridge}--\ref{prop:halforder-bridge-stability} yield a practical principle:
\emph{use the geometric/half-order bridge as a robust baseline when overlap or tails are uncertain,
and use the half-order overlap $\rho$---estimable with bounded variance under either model---as a stable tuning target
when constructing references or intermediate bridging
 schedules.}
Concretely, when a difficult ratio is decomposed into a product of easier ratios, a natural operational goal is
to ensure that the pairwise overlaps $\rho$ between consecutive states are not too small.
This overlap-guided viewpoint is pursued further for IS weight diagnostics in Section~\ref{subsec:is-degeneracy}.


\paragraph*{Simulation study 3: Half-order bridge vs.\ one-sided importance sampling under tail mismatch.}
To complement Propositions~\ref{prop:halforder-bridge}--\ref{prop:halforder-bridge-stability}, we give a minimal
toy family in which \emph{either} one-sided normalizing-constant estimator can be theoretically ill-posed
(infinite variance) depending on the direction of tail mismatch, while the \emph{half-order} bridge estimator
remains stable in both directions.

\emph{Model pair and the target ratio.}
On $(0,1)$, let $\tilde p_2(x)\equiv 1$ so that $Z_2=\int_0^1 \tilde p_2(x)\,dx=1$ and $p_2(x)\equiv 1$.
Let $p_1$ be the $\mathrm{Beta}(a,1)$ density, $p_1(x)=a\,x^{a-1}$, and set $\tilde p_1(x)=r\,p_1(x)$ so that
$Z_1=\int_0^1 \tilde p_1(x)\,dx=r$. Then the normalizing-constant ratio is $Z_1/Z_2=r$ and the unnormalized ratio is
\[
w(x):=\frac{\tilde p_1(x)}{\tilde p_2(x)} = r\,p_1(x)=r\,a\,x^{a-1}.
\]
(For the Monte Carlo experiments below we set $r=1$ without loss of generality, since we report relative error.)

\emph{Estimators and cost matching.}
We compare the half-order bridge estimator
\[
\widehat r_{1/2}
:=\frac{\widehat a_2}{\widehat a_1},
\quad
\widehat a_2:=\frac{1}{n_2}\sum_{i=1}^{n_2} w(X_{2i})^{1/2},
\quad
\widehat a_1:=\frac{1}{n_1}\sum_{j=1}^{n_1} w(X_{1j})^{-1/2},
\]
with independent samples $X_{2i}\stackrel{\mathrm{iid}}{\sim}p_2$ and $X_{1j}\stackrel{\mathrm{iid}}{\sim}p_1$,
against the appropriate one-sided estimator in each stress test:
\[
\widehat r_{\mathrm{F}}
:=\frac{1}{N}\sum_{i=1}^{N} w(X_{2i}), \qquad X_{2i}\stackrel{\mathrm{iid}}{\sim}p_2
\quad\text{(forward one-sided IS; $t=1$)},
\]
and
\[
\widehat r_{\mathrm{R}}
:=\left(\frac{1}{N}\sum_{i=1}^{N} w(X_{1i})^{-1}\right)^{-1}, \qquad X_{1i}\stackrel{\mathrm{iid}}{\sim}p_1
\quad\text{(reverse one-sided IS; based on $t=0$).}
\]
To make comparisons cost-matched, we fix the total number of ratio evaluations to $N_{\text{total}}=2000$.
For a one-sided estimator we take $N=2000$ draws from its sampling distribution.
For the half-order bridge we split $(n_1,n_2)=(1000,1000)$.

\emph{Why this family is diagnostic: forward and reverse breakdown.}
The forward estimator $\widehat r_{\mathrm{F}}$ is unbiased ($\E_{p_2}[w]=r$) but its variance can be infinite:
\[
\E_{p_2}[w^2]
= r^2 \int_0^1 p_1(x)^2\,dx
= r^2 a^2 \int_0^1 x^{2a-2}\,dx
=
\infty
\quad\text{whenever } a\le \tfrac12 .
\]
Conversely, the inverse of the reverse estimator, $\hat{r}_R^{-1}$, is unbiased for $r^{-1}$ (since $\E_{p_1}[w^{-1}]=r^{-1}$)
, but it can be ill-posed
because
\[
\E_{p_1}[w^{-2}]
= r^{-2}\int_0^1 p_1(x)^{-1}\,dx
= \frac{r^{-2}}{a}\int_0^1 x^{1-a}\,dx
=
\infty
\quad\text{whenever } a\ge 2 .
\]
Thus, one-sided ratio estimation can fail on \emph{either side}, depending on the direction of tail mismatch.

\emph{Half-order prediction via overlap.}
In this $\mathrm{Beta}(a,1)$ vs.\ $\mathrm{Unif}(0,1)$ family, the half-order overlap is available in closed form:
\[
\rho(a)=\int_0^1 \sqrt{p_1(x)p_2(x)}\,dx
= \int_0^1 \sqrt{a}\,x^{(a-1)/2}\,dx
=\frac{2\sqrt{a}}{a+1},
\qquad
\rho(a)^2=\frac{4a}{(a+1)^2}.
\]
Proposition~\ref{prop:halforder-bridge-stability} yields the asymptotic relative standard deviation
\[
\mathrm{RSD}(\widehat r_{1/2})
\approx
\sqrt{\frac{1-\rho^2}{\rho^2}\Bigl(\frac{1}{n_1}+\frac{1}{n_2}\Bigr)}
=
\sqrt{\bigl(\rho^{-2}-1\bigr)\Bigl(\frac{1}{n_1}+\frac{1}{n_2}\Bigr)}.
\]
With $n_1=n_2=1000$, this predicts $\mathrm{RSD}(\widehat r_{1/2})\approx 1.58\%$ for $a=\tfrac12$
and $\approx 2.58\%$ for $a=3$.

\begin{figure}[!htbp]
\centering
\includegraphics[width=0.98\linewidth]{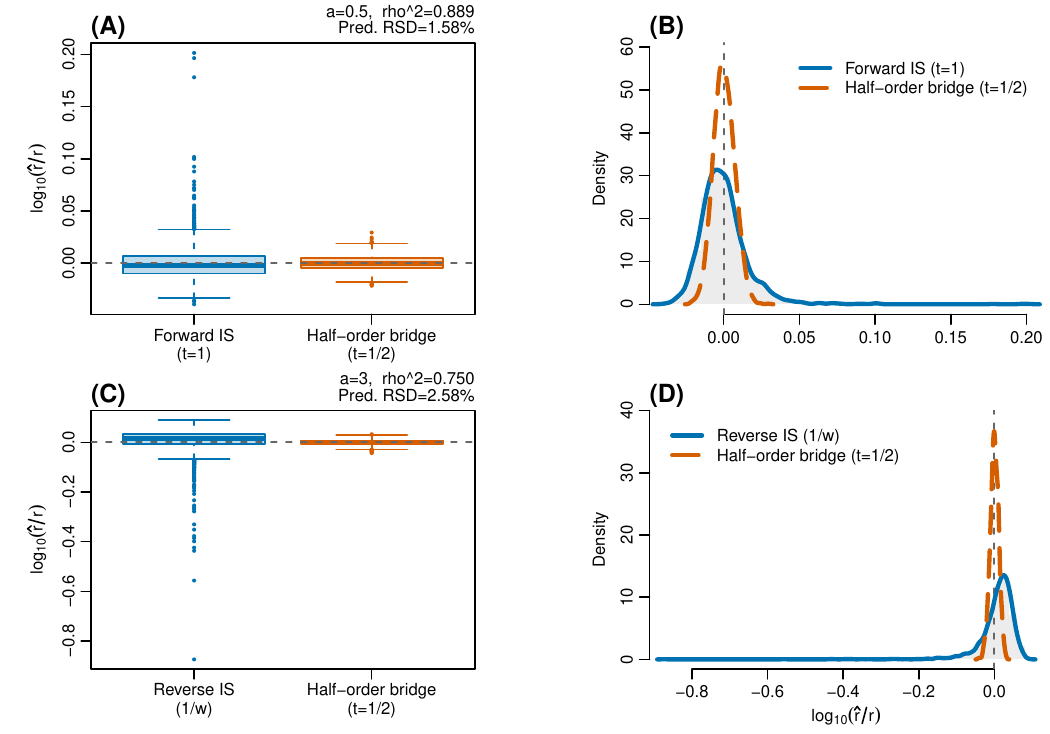}
\caption{Toy normalizing-constant ratio estimation under tail mismatch in the $\mathrm{Beta}(a,1)$ vs.\ $\mathrm{Unif}(0,1)$ family (Simulation study~3).
We report $\log_{10}(\widehat r/r)$ over $R=2000$ independent repetitions at matched computational cost $N_{\text{total}}=2000$.
\textbf{Top row ($a=\tfrac12$): forward breakdown.}
Panels (A)--(B) compare the forward one-sided estimator $\widehat r_{\mathrm{F}}=\frac1N\sum_{i=1}^N w(X_{2i})$ ($X_{2i}\sim p_2$; $N=2000$)
to the half-order bridge $\widehat r_{1/2}$ ($n_1=n_2=1000$).
Because $\E_{p_2}[w^2]=\infty$ at $a=\tfrac12$, the forward estimator exhibits a pronounced right tail and many outliers.
\textbf{Bottom row ($a=3$): reverse breakdown.}
Panels (C)--(D) compare the reverse one-sided estimator $\widehat r_{\mathrm{R}}=\{\frac1N\sum_{i=1}^N w(X_{1i})^{-1}\}^{-1}$ ($X_{1i}\sim p_1$; $N=2000$)
to the same half-order bridge estimator ($n_1=n_2=1000$).
Because $\E_{p_1}[w^{-2}]=\infty$ for $a\ge 2$, the reverse estimator develops a heavy left tail (extreme underestimation on the log scale).
In all panels, the half-order bridge remains tightly concentrated around $0$.
Panel headers report the closed-form overlap $\rho^2=4a/(a+1)^2$ and the corresponding asymptotic RSD prediction from
Proposition~\ref{prop:halforder-bridge-stability}.}
\label{fig:toy-halforder-bridge}
\end{figure}

\emph{Takeaway.}
This toy family isolates the practical content of the half-order theory for normalizing-constant ratios:
one-sided ratio estimators can be dominated by rare extreme weights in \emph{either direction} of tail mismatch,
making CLT-based error assessment ill-founded, whereas the geometric/half-order bridge provides a conservative default
whose Monte Carlo building blocks admit distribution-free second-moment control and whose relative error is governed by
the stably estimable overlap $\rho$.

\subsection{Importance sampling view: degeneracy diagnostics without fragile second moments}\label{subsec:is-degeneracy}

\paragraph*{Importance sampling as ratio estimation and moment fragility.}
Importance sampling (IS) approximates expectations under a target density $p_1$ by reweighting draws from
a proposal density $p_2$.  Its numerical behavior is controlled by the density ratio
\[
W(x):=\frac{p_1(x)}{p_2(x)},
\]
which is the same ratio object that appears in the Turing--Good identities of
Section~\ref{sec:TG-variance-halforder}.
A key practical issue is that many widely used diagnostics and tuning criteria for IS
(e.g., squared-weight concentration summaries, coefficients of variation, and other moment-based rules)
are implicitly functions of second or higher moments of $W$ (or power transforms $W^t$).
These moments can be extremely unstable or even infinite under mild tail mismatch, which makes such criteria
ill-posed precisely in the regimes where they are most needed.

Section~\ref{sec:TG-variance-halforder} implies a simple but consequential principle:
among power transforms of the ratio, the half-order point $t=\tfrac12$ is the \emph{unique} choice that
admits a distribution-free worst-side second-moment guarantee.
This motivates using the half-order overlap $\rho =\int \sqrt{p_1(x)p_2(x)}\,d\mu$ 
as a conservative, stably estimable summary of weight concentration, and as an anchor for proposal and
bridging 
design.

\begin{proposition}[Half-order overlap and coefficient of variation as a stable degeneracy summary]
\label{prop:halforder-ess}
Let $p_1$ and $p_2$ be mutually absolutely continuous and define the IS weight
$W(X):=p_1(X)/p_2(X)$ together with its half-order transform
\[
U(X):=W(X)^{1/2}=\sqrt{\frac{p_1(X)}{p_2(X)}},
\qquad
\rho:=\int \sqrt{p_1(x)p_2(x)}\,d\mu.
\]
Then, under $p_2$,
\[
\E_{p_2}[U]=\rho,
\qquad
\Var_{p_2}(U)=1-\rho^2\le 1,
\]
and symmetrically, under $p_1$, $\E_{p_1}[W^{-1/2}]=\rho$ and $\Var_{p_1}(W^{-1/2})=1-\rho^2$.
Consequently, the half-order squared coefficient of variation is
\[
\CV_{1/2}^2
:=\frac{\Var_{p_2}(U)}{\E_{p_2}[U]^2}
=\frac{1-\rho^2}{\rho^2}
=\rho^{-2}-1.
\]
In particular, $\rho^2=1/(1+\CV_{1/2}^2)$ provides an interpretable overlap index in $(0,1]$
that remains well-defined and stably estimable under either model.
\end{proposition}

\begin{proof}
The mean identities are the half-order instance of the Turing--Good identity (Theorem~\ref{thm:TG})
applied to $B=W$; the variance statements are Corollary~\ref{cor:half-order-eq}. The CV identity follows by algebra.
\end{proof}

\paragraph*{Remark (Estimating $\rho$ and $\rho^2$ from one or two sides).}
Proposition~\ref{prop:halforder-ess} directly yields bounded-variance estimators of $\rho$.
If $X_1,\dots,X_m\stackrel{\mathrm{iid}}{\sim}p_2$ and the normalized ratio $W=p_1/p_2$ is available, then
\[
\widehat\rho_{(2)}:=\frac{1}{m}\sum_{i=1}^m \sqrt{W(X_i)}
\qquad\text{satisfies}\qquad
\Var(\widehat\rho_{(2)})=\frac{1-\rho^2}{m}\le \frac{1}{m}.
\]
Likewise, if $Y_1,\dots,Y_m\stackrel{\mathrm{iid}}{\sim}p_1$, then
\[
\widehat\rho_{(1)}:=\frac{1}{m}\sum_{i=1}^m W(Y_i)^{-1/2}
\qquad\text{satisfies}\qquad
\Var(\widehat\rho_{(1)})=\frac{1-\rho^2}{m}\le \frac{1}{m}.
\]
In many Bayesian applications $p_1$ is only known up to a normalizing constant, so one instead works with an
unnormalized ratio $w=\tilde p_1/\tilde p_2 = r\,W$.
In that case, the half-order building blocks from Section~\ref{subsec:normconst} still allow stable overlap estimation:
by Proposition~\ref{prop:halforder-bridge},
$\E_{p_2}[w^{1/2}]=r^{1/2}\rho$ and $\E_{p_1}[w^{-1/2}]=r^{-1/2}\rho$, hence
\begin{equation}\label{eq:rho2-product}
\rho^2
=\E_{p_2}\!\left[w^{1/2}\right]\;\E_{p_1}\!\left[w^{-1/2}\right].
\end{equation}
Accordingly, with samples from both sides one may estimate $\rho^2$ by the product of the corresponding sample means.
In particular, if $\widehat a_2$ and $\widehat a_1$ denote the sample means of $w^{1/2}$ under $p_2$ and $w^{-1/2}$ under $p_1$
(as in Proposition~\ref{prop:halforder-bridge-stability}), then
\[
\widehat\rho^2:=\widehat a_2\,\widehat a_1,
\qquad
\widehat\rho:=(\widehat a_2\,\widehat a_1)^{1/2},
\]
is available alongside $\widehat r_{1/2}=\widehat a_2/\widehat a_1$ at essentially no additional cost
(see also the overlap remark in Section~\ref{subsec:normconst}).

\begin{proposition}[Uniqueness: power-moment degeneracy diagnostics can be ill-posed away from $t=\tfrac12$]
\label{prop:is-fragility}
Fix $t\neq \tfrac12$ and consider the power-weight transform $W^t$.
There exist mutually absolutely continuous pairs $(p_1,p_2)$ for which
$\Var_{p_2}(W^t)=\infty$ (if $t>\tfrac12$) or $\Var_{p_1}(W^{t-1})=\infty$ (if $t<\tfrac12$).
In particular, diagnostics and tuning rules that require finite second moments of $W^t$
can be ill-posed away from the half-order point even when $p_1$ and $p_2$ are valid IS pairs.
\end{proposition}

\begin{proof}
Apply Theorem~\ref{thm:minimax}(ii) to the Bayes factor $B=W$.
\end{proof}

\paragraph*{Practical recipe: overlap-guided diagnostics and design.}
The above results suggest using the half-order overlap as a conservative anchor in regimes where tail mismatch is uncertain.
\begin{enumerate}
\item \emph{Report a half-order overlap estimate.}
Compute $\widehat\rho_{(2)}$ from proposal draws (or $\widehat\rho_{(1)}$ from target draws) when $W$ is available,
and otherwise use the two-sided product estimator implied by~\eqref{eq:rho2-product} when working with unnormalized ratios.
Because the corresponding summands have bounded variance, this diagnostic remains well-posed under worst-case tail mismatch.

\item \emph{Interpretation.}
Small $\widehat\rho$ (equivalently, large $\widehat{\CV}_{1/2}$) indicates severe weight concentration
and warns against reliance on tail-sensitive, higher-moment diagnostics.

\item \emph{Design of proposals and intermediate schedules.}
When a difficult reweighting problem is decomposed into a product of easier steps
(tempering/bridging/annealing), use pairwise half-order overlap as a tuning target:
estimate overlaps between consecutive states, and refine the grid (or adjust proposals) until adjacent overlaps
are not excessively small.  This directly aligns design with an overlap quantity that remains estimable
in worst-case tail regimes.
\end{enumerate}



\paragraph*{Simulation study: weight concentration away from $t=\tfrac12$ and stability anchored by the half-order overlap.}
We close this subsection with a simple numerical illustration of
Propositions~\ref{prop:halforder-ess}--\ref{prop:is-fragility}.
The goal is to visualize two points:
(i) power-weight diagnostics away from $t=\tfrac12$ can become ill-posed when the relevant second moments do not exist,
and (ii) the half-order overlap $\rho$ provides a stable, interpretable benchmark for weight concentration that remains
well-defined under worst-case tail mismatch.

\smallskip
\noindent\textit{Model pair.}
We use the one-parameter family on $(0,1)$ that also appears in the proof of
Theorem~\ref{thm:minimax}(ii): let $p_2(x)=1$ (uniform) and $p_1(x)=a x^{a-1}$ (a $\mathrm{Beta}(a,1)$ density),
so that the importance weight is
\[
W(x)\;=\;\frac{p_1(x)}{p_2(x)}\;=\;a x^{a-1}.
\]
In this family, the half-order overlap admits a closed form,
\[
\rho=\int_0^1 \sqrt{p_1(x)p_2(x)}\,dx=\frac{2\sqrt{a}}{a+1},
\qquad
\rho^2=\frac{4a}{(a+1)^2}.
\]
Proposition~\ref{prop:halforder-ess} implies that half-order summaries based on $\sqrt{W}$ (or $1/\sqrt{W}$)
remain stable and that $\rho^2$ plays the role of a deterministic overlap benchmark.

\smallskip
\noindent\textit{Weight transforms and concentration summaries.}
For a set of $N$ draws $\{X_i\}_{i=1}^N$ from a designated generating distribution, we compute unnormalized weights
$w_i$ and normalize them as $\tilde w_i = w_i/\sum_{j=1}^N w_j$.
We report three complementary summaries of concentration:
\[
C(p):=\sum_{i=1}^{\lceil pN\rceil}\tilde w_{(i)},\quad p\in(0,1],
\qquad
\kappa_N:=\frac{1}{N\sum_{i=1}^N \tilde w_i^2},
\qquad
S_{0.01}:=\sum_{i=1}^{\lceil 0.01N\rceil} \tilde w_{(i)},
\]
where $\tilde w_{(1)}\ge \cdots \ge \tilde w_{(N)}$ are the normalized weights in decreasing order.
Here $C(p)$ is a Lorenz-type concentration curve (uniform weights correspond to $C(p)=p$),
$\kappa_N$ is a normalized reciprocal squared-weight concentration index
(often reported in practice as $\mathrm{ESS}/N$), and $S_{0.01}$ measures the cumulative mass
carried by the top $1\%$ of weights.

\smallskip
\noindent\textit{Two stress tests away from $t=\tfrac12$.}
Figure~\ref{fig:is_halforder} summarizes the behavior over $R=500$ independent replicates at fixed $N=2000$.
In the proposal-side stress test (top row) we draw $X_i\sim p_2$ and compare the naive weights $w=W$ ($t=1$)
to half-order weights $w=W^{1/2}$. With the boundary choice $a=1/2$, we have $\Var_{p_2}(W)=+\infty$
(Proposition~\ref{prop:is-fragility}), and the resulting normalized weights exhibit extreme concentration.
In the target-side stress test (bottom row) we draw $X_i\sim p_1$ and compare a representative exponent below
half order (here $w=W^{t-1}$ with $t=1/4$) to the half-order reverse weights $w=W^{-1/2}$.
With the boundary choice $a=3$, we have $\Var_{p_1}(W^{t-1})=+\infty$ for this exponent, again producing
severe concentration.

\smallskip
\noindent\textit{Half-order stabilization and overlap benchmarks.}
In both stress tests, the half-order transforms yield mild concentration: the Lorenz curves lie close to the
uniform baseline, the top-share $S_{0.01}$ remains near $0.01$, and the squared-weight index $\kappa_N$ concentrates
near the deterministic benchmark $\rho^2=4a/(a+1)^2$.
This illustrates the design message of Section~\ref{subsec:is-degeneracy}:
diagnostics and tuning rules anchored at $t=\tfrac12$ (equivalently, based on $\sqrt{W}$ or $1/\sqrt{W}$)
remain well-posed under worst-case tail mismatch, whereas power-weight rules away from $t=\tfrac12$ can become
fragile or ill-defined even in mutually absolutely continuous settings.

\begin{figure}[!htbp]
  \centering
  \includegraphics[width=.98\linewidth]{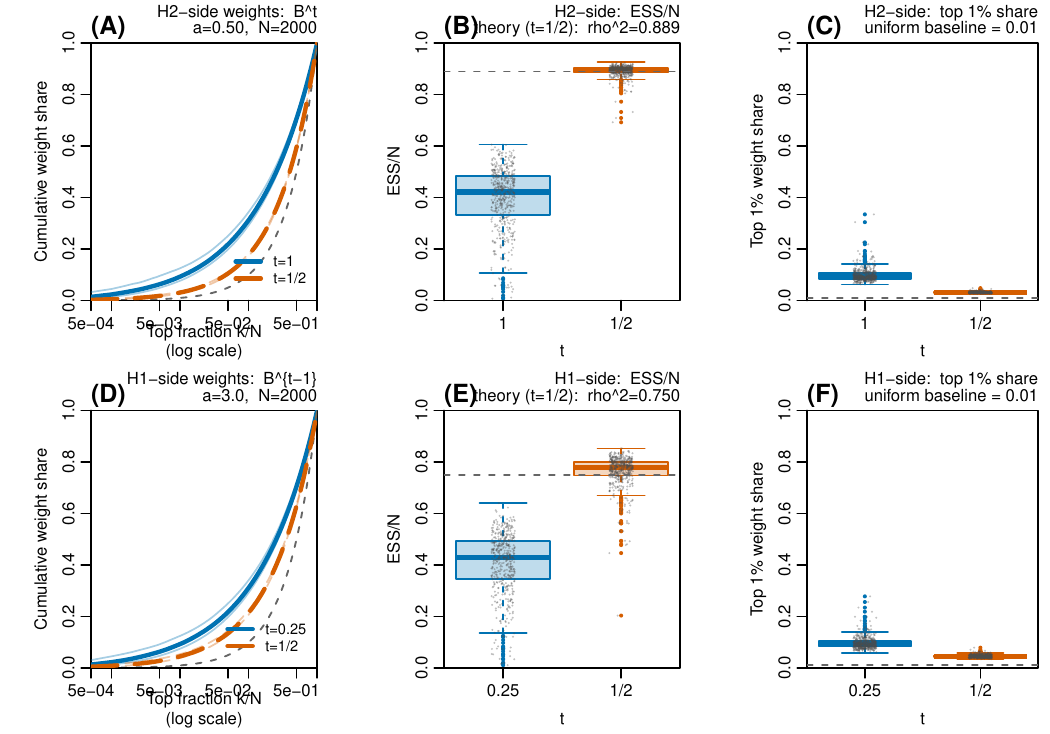}
  \caption{Importance-sampling weight concentration away from $t=\tfrac12$ and stability anchored by the half-order overlap
  in the $\mathrm{Beta}(a,1)$ vs.\ $\mathrm{Unif}(0,1)$ family ($R=500$ replicates, $N=2000$).
  Top row (proposal-side): $X_i\sim p_2$ with $a=1/2$, comparing weights $w=W^t$ for $t\in\{1,\tfrac12\}$.
  Bottom row (target-side): $X_i\sim p_1$ with $a=3$, comparing reverse weights $w=W^{t-1}$ for $t\in\{\tfrac14,\tfrac12\}$.
  Left panels show Lorenz-type concentration curves $C(p)=\sum_{i=1}^{\lceil pN\rceil}\tilde w_{(i)}$
  (median and 10/90\% envelopes across replicates) on a log-$p$ scale; the diagonal $C(p)=p$ corresponds to uniform weights.
  Middle panels show boxplots of the normalized reciprocal squared-weight concentration
  $\kappa_N := 1/(N\sum_{i=1}^N \tilde w_i^2)$ (commonly reported as $\mathrm{ESS}/N$), with dotted reference lines at the
  half-order overlap benchmark $\rho^2=4a/(a+1)^2$.
  Right panels show the cumulative weight share carried by the largest 1\% of weights ($S_{0.01}$);
  the dotted line is the uniform baseline $0.01$.
  In both stress tests, the half-order weights exhibit mild concentration and $\kappa_N$ close to $\rho^2$,
  whereas the tail-sensitive exponents yield substantial concentration and degraded squared-weight summaries.}
  \label{fig:is_halforder}
\end{figure}

\section{Discussion}
\label{sec:discussion}

This paper revisits the classical Turing--Good identities through a deliberately
computation-aware lens.  While the identities themselves are exact population equalities
that hold for a continuum of exponents, their usefulness as \emph{numerical} checks and as
building blocks for Monte Carlo estimators is governed by higher-order behavior of the
Bayes factor 
In particular, in
heavy-tailed regimes with weak overlap---precisely where algorithmic validation is most
needed---the integer-order moments traditionally used for diagnostics can have enormous or
even infinite variance, so that the check itself becomes ill-posed.  Our main message is
that the half-order point $t=1/2$ is not merely a convenient fractional choice but a
structurally singled-out exponent: it is the unique choice that yields two-sided,
distribution-free stability in second moments.

\paragraph*{Implications for Good checks in Bayesian workflows.}
A recurring challenge in Bayes factor computation is that correctness is difficult to
validate from a single realized dataset: the Bayes factor is a ratio of marginal
likelihoods, and numerical approximations can fail silently when tails are mismatched or
when the implementation inadvertently differs from the intended prior specification.
Good checks \citep{Sekulovski2024} address this by comparing empirical Monte Carlo moments
against exact identities that hold under prior predictive simulation.  Our results refine
the design principle for such checks.

The half-order viewpoint yields two practical advantages.  First, \emph{symmetry}: the
two-sided half-order check does not require a judgement about which model is ``true'' or
``more complex.''  Because $(p_1/p_2)^{-1/2}=(p_2/p_1)^{1/2}$ pointwise, the same
theoretical target $\rho$ is approached from either generating model, and the diagnostic
remains well-posed whichever direction is computationally convenient.  Second,
\emph{worst-case stability}: each summand in the check has variance at most~1, and the
balanced two-sample discrepancy has variance $4(1-\rho^2)/N\le 4/N$ under a fixed budget
of $N$ Bayes factor evaluations (Theorem~\ref{thm:unified}).  Moreover, in the small-overlap
regime $\rho\le 1/2$, the balanced two-sample half-order check can be more efficient than
the one-sided $t=1$ Turing check at matched cost (Theorem~\ref{thm:unified}(ii)).  This is
qualitatively different from one-sided integer-moment checks, which can have arbitrarily
large or infinite variance for valid (mutually absolutely continuous) model pairs.

From a workflow perspective, this suggests a conservative default: when one can simulate
from both prior predictives, start with a balanced two-sample half-order check and report
(i) the discrepancy $\widehat{\Delta}=\widehat{\rho}_2-\widehat{\rho}_1$ together with a
studentized standard error and (ii) the pooled overlap estimate
$\widehat{\rho}=(\widehat{\rho}_1+\widehat{\rho}_2)/2$ as an interpretable measure of
model overlap.  Large systematic discrepancies indicate implementation, prior, or
simulation mismatches, while a small $\widehat{\rho}$ signals an intrinsically hard
Bayes factor problem in which tail sensitivity and weight degeneracy are expected.

The appearance of $\rho$ is not incidental: it is simultaneously (i) the half-order
Hellinger overlap that always exists and lies in $(0,1]$ and (ii) the only model-dependent
quantity governing the bounded half-order variance.  Thus, $\rho$ plays a dual role in
practice.  As a diagnostic target, it is a quantity that can be estimated from either
side with uniformly controlled variance.  As a difficulty index, it quantifies how
strongly separated the marginal likelihoods are.  When $\rho$ is small, one should expect
heavy-tailed Bayes factors and unstable one-sided estimators based on raw weights; this
is precisely the regime in which the simulation studies show dramatic failures of
integer-moment checks.

More broadly, expressing algorithmic stability in terms of $\rho$ invites a useful shift
in how we communicate Bayes factor computations.  Rather than reporting only the Bayes
factor and an algorithm-specific Monte Carlo standard error, it can be informative to
also report an overlap statistic (e.g.\ $\widehat{\rho}$ or its complement
$1-\widehat{\rho}^2$) that is robustly estimable and that directly signals potential
weight degeneracy.  This complements established recommendations for principled Bayesian
workflows that emphasize diagnostics and validation alongside inference (e.g.,
\citealt{SchadEtAl2023}).

\paragraph*{Connections to ratio estimation and Monte Carlo design.}
Although we motivated the analysis via Good checks, the underlying structure is about
density ratios and normalizing-constant estimation.  Section~\ref{sec:broader-implications}
highlights that the same half-order midpoint $\sqrt{p_1 p_2}$ naturally produces stable
building blocks for ratio estimation (Propositions~\ref{prop:halforder-bridge} and
\ref{prop:halforder-bridge-stability}).  From the perspective of bridge sampling and
free-energy methods (Bennett 1976; Meng \& Wong 1996; Shirts \& Chodera 2008; Gronau
et al.\ 2017), the half-order identity may be viewed as a geometric ``baseline bridge''
whose components admit a distribution-free second-moment guarantee under mutual absolute
continuity.  Similarly, Section~\ref{subsec:is-degeneracy} frames importance-sampling
degeneracy diagnostics in terms of the same overlap family: diagnostics based on
$E[W^2]$ can be ill-posed in legitimate importance-sampling problems, whereas the
half-order transform yields an always-finite and interpretable overlap/degeneracy index
(Proposition~\ref{prop:halforder-ess}).  
Together, these connections suggest that half-order overlap can serve as a common
currency across Bayes factor computation, ratio estimation and proposal/bridge design. 

\paragraph*{Limitations and directions for future work.}
Our minimax stability statement is deliberately worst-case and focuses on second moments.
There are several natural extensions.
First, many practical Bayes factor computations rely on MCMC rather than i.i.d.\ draws
from $p_1$ and $p_2$.  While the half-order \emph{per-draw} variance bound remains
relevant, dependence inflates Monte Carlo error via autocorrelation, and a full analysis
should incorporate effective sample sizes and Markov-chain CLTs.  Developing half-order
Good-check calibration rules that remain reliable under MCMC dependence is an important
practical direction.
Second, although half order guarantees finite variance uniformly, it does not guarantee
small variance: when $\rho$ is extremely small, $1-\rho^2\approx 1$ and the two-sided
check may still require substantial Monte Carlo effort to achieve tight tolerances.  In
such regimes, the overlap estimate itself is useful as a warning signal that additional
algorithmic measures are needed (e.g., bridging
with intermediate
distributions).  A constructive design theory that uses estimated half-order overlaps to
adaptively place bridges (for example, choosing consecutive distributions to maintain a
roughly constant pairwise $\rho$) would operationalize the ideas in
Section~\ref{sec:broader-implications}.
Third, our analysis assumes mutual absolute continuity so that the Bayes factor is
well-defined almost everywhere.  In some applied settings (e.g., models with boundary
constraints, discrete--continuous mixtures, or hard truncations), this assumption may fail
or hold only approximately.  Extending half-order diagnostics to such settings---for
instance via truncation, regularization, or partial-overlap decompositions---would broaden
applicability.


\paragraph*{Concluding perspective.}
The Turing--Good identities provide a rare bridge between exact measure-theoretic
relationships and practical computational diagnostics.  The present work shows that this
bridge becomes particularly sturdy at the half-order point: $t=1/2$ is the unique exponent
for which the identity can be checked in a symmetric, two-sided, worst-case stable manner,
and the associated overlap $\rho$ emerges as a robust summary of both model similarity and
Monte Carlo difficulty.  We hope that framing Bayes factor validation and ratio-estimation
design around half-order overlap will help make Bayes factor workflows more reliable,
transparent, and reproducible in the regimes where they are most challenging.












  \bibliography{bibliography.bib}

\appendix




\end{document}

%% file: figs/Table1_1A_rev2.tex
\begin{table}[t]
\centering
\setlength{\tabcolsep}{4pt}
\begin{tabular}{lccccc}
\toprule
& $t=\tfrac12$ & \multicolumn{2}{c}{$t=1$} & \multicolumn{2}{c}{$t=2$} \\
\cmidrule(lr){2-2}\cmidrule(lr){3-4}\cmidrule(lr){5-6}
\, $n$ & $\widehat\Delta_{1/2}$ & $\widehat{\E}_{\Htwo}[B]-1$ & $\widehat{\E}_{\Hone}[B^{-1}]-1$ & $\widehat{\E}_{\Htwo}[B^{2}]-\widehat{\E}_{\Hone}[B]$ & $\widehat{\E}_{\Hone}[B^{-2}]-\widehat{\E}_{\Htwo}[B^{-1}]$ \\
\midrule
$10$ & $0.000\,(0.018)$ & $-0.001\,(0.095)$ & $0.000\,(0.022)$ & $-0.074\,(8.508)$ & $0.000\,(0.059)$ \\
$50$ & $0.000\,(0.023)$ & $-0.209\,(7.278)$ & $0.000\,(0.039)$ & $-8.84\mathrm{e}{+}11\,(9.62\mathrm{e}{+}10)$ & $-0.001\,(0.198)$ \\
$100$ & $0.000\,(0.023)$ & $-0.441\,(3.103)$ & $0.000\,(0.048)$ & $-2.51\mathrm{e}{+}26\,(3.91\mathrm{e}{+}25)$ & $0.000\,(0.340)$ \\
\bottomrule
\end{tabular}
\caption{Simulation study~1A (binomial example, correctly specified): Monte Carlo \emph{differences}. Entries are Mean (SD) over $R=10000$ repetitions, each using $m=2000$ draws per hypothesis.}
\label{tab:sim1_binom_correct}
\end{table}

%% file: figs/Table1_1B_rev2.tex
\begin{table}[t]
\centering
\setlength{\tabcolsep}{4pt}
\begin{tabular}{lccccc}
\toprule
& $t=\tfrac12$ & \multicolumn{2}{c}{$t=1$} & \multicolumn{2}{c}{$t=2$} \\
\cmidrule(lr){2-2}\cmidrule(lr){3-4}\cmidrule(lr){5-6}
\, $n$ & $\widehat\Delta_{1/2}$ & $\widehat{\E}_{\Htwo}[B]-1$ & $\widehat{\E}_{\Hone}[B^{-1}]-1$ & $\widehat{\E}_{\Htwo}[B^{2}]-\widehat{\E}_{\Hone}[B]$ & $\widehat{\E}_{\Hone}[B^{-2}]-\widehat{\E}_{\Htwo}[B^{-1}]$ \\
\midrule
$10$ & $-0.051\,(0.017)$ & $0.000\,(0.094)$ & $0.079\,(0.022)$ & $3.213\,(8.465)$ & $0.171\,(0.059)$ \\
$50$ & $-0.061\,(0.025)$ & $-0.109\,(10.188)$ & $0.108\,(0.041)$ & $-5.43\mathrm{e}{+}11\,(7.62\mathrm{e}{+}10)$ & $0.447\,(0.210)$ \\
$100$ & $-0.054\,(0.054)$ & $4.329\,(474.266)$ & $0.112\,(0.050)$ & $-1.35\mathrm{e}{+}26\,(2.86\mathrm{e}{+}25)$ & $0.644\,(0.358)$ \\
\bottomrule
\end{tabular}
\caption{Simulation study~1B (binomial example, correctly specified): Monte Carlo \emph{differences}. Entries are Mean (SD) over $R=10000$ repetitions, each using $m=2000$ draws per hypothesis.}
\label{tab:sim1_binom_mismatch}
\end{table}